\documentclass[a4paper, reqno, 14pt]{amsart}
\usepackage[utf8]{inputenc}

\usepackage{amsthm,amsfonts,amssymb,amsmath, microtype}
\usepackage[all]{xy}
\usepackage{xr-hyper}
\usepackage[colorlinks=true]{hyperref}
\usepackage{verbatim}
\usepackage{pdfsync}

\usepackage{tikz}
\usetikzlibrary{arrows,calc,positioning}
\usetikzlibrary{circuits.logic.IEC}
\usepackage{xcolor}
\usetikzlibrary{shapes}
\usetikzlibrary{decorations.markings}

\tikzstyle heightone=[scale=.7,shift={(0,-.3)}]
\tikzstyle heightones=[scale=.8,xscale=.35,shift={(0,.1)}]
\tikzstyle heightoneonehalf=[scale=.9,shift={(0,-.2)}]
\tikzstyle heighttwo=[scale=.9,shift={(0,-.4)}]
\tikzstyle heighttwos=[scale=.5,xscale=.6,shift={(0,-.1)}]
\tikzstyle heightthree=[scale=.6,shift={(0,-.9)}]
\tikzstyle heightthrees=[scale=.4,xscale=.7,shift={(0,-.2)}]

\tikzstyle arrowstyle=[blue,semitransparent,scale=2]

\tikzstyle basiclabel=[draw=none,fill=none,shape=rectangle,inner sep=2pt,scale=.8]
\tikzstyle leftlabel=[basiclabel,anchor=east]
\tikzstyle rightlabel=[basiclabel,anchor=west]
\tikzstyle bottomlabel=[basiclabel,anchor=north]
\tikzstyle toplabel=[basiclabel,anchor=south]

\tikzstyle vertex=[circle,draw,fill=black,inner sep=1pt]
\tikzstyle ciliation=[circle,draw=none,fill=red,inner sep=1pt,semitransparent]
\tikzstyle ciliatednode=[vertex,pin={[pin distance=1mm,pin edge={semitransparent,red},ciliation]#1:{}}]

\tikzstyle vector=[black,thick,rectangle,draw=gray!50!yellow,top color=yellow!30,bottom color=black!10,scale=.8,inner sep=2pt]
\tikzstyle small vector=[vector,scale=.8]
\tikzstyle plain vector=[rectangle,draw=none,fill=white,scale=.7]

\tikzstyle my signal=[black,thick,signal,signal pointer angle=120,draw=blue!50,top color=blue!20,bottom color=black!10,scale=.8,inner sep=2pt]
\tikzstyle matrix=[my signal,signal from=south,signal to=north]
\tikzstyle reverse matrix=[my signal,signal from=north,signal to=south]

\tikzstyle small matrix=[matrix,scale=.7]
\tikzstyle reverse small matrix=[reverse matrix,scale=.7]
\tikzstyle matrix on edge=[small matrix,sloped,rotate=-90]
\tikzstyle reverse matrix on edge=[small matrix,sloped,rotate=90]

\tikzstyle trivalent=[very thick]
\tikzstyle dotdotdot=[decorate,decoration={markings,
    mark=at position .3 with{\node{.};},
    mark=at position .5 with {\node{.};},
    mark=at position .7 with {\node{.};}}]

\tikzstyle wavyup=[out=90,in=-90]
\tikzstyle wavydown=[out=-90,in=90]

\tikzstyle symmetrizer=[rectangle,fill=gray!10,draw=black]
\tikzstyle permutation=[symmetrizer]
\tikzstyle antisymmetrizer=[rectangle,fill=black,draw=black]
\tikzstyle symlabel=[draw=none,fill=none,black,scale=.8]
\tikzstyle asymlabel=[draw=none,fill=none,white,scale=.8]

\newcommand{\BC}{{\mathbb {C}}}

\newcommand{\BE}{{\mathbb {E}}}

\newcommand{\BN}{{\mathbb {N}}}

\newcommand{\BP}{{\mathbb {P}}}

\newcommand{\BR}{{\mathbb {R}}}
\newcommand{\BS}{{\mathbb {S}}}

\newcommand{\Sym}{\mathrm{Sym}}
\newcommand{\polylog}{\mathrm{polylog}}

\newcommand{\CA}{{\mathcal {A}}}

\newcommand{\CD}{{\mathcal {D}}}

\newcommand{\CH}{{\mathcal {H}}}

\newcommand{\CM}{{\mathcal {M}}}
\newcommand{\CN}{{\mathcal {N}}}
\newcommand{\CO}{{\mathcal {O}}}

\newcommand{\CS}{{\mathcal {S}}}

\newcommand{\rank}{{\mathrm{rank}}}

\newcommand{\sgn}{{\mathrm{sgn}}}

\newcommand{\tr}{{\mathrm{tr}}}

\newcommand{\id}{{\mathrm{id}}}

\newcommand{\dist}{{\mathrm{dist}}}

\newtheorem{theorem}{Theorem}
\newtheorem{proposition}[theorem]{Proposition}
\newtheorem{lemma}[theorem]{Lemma}

\theoremstyle{definition}
\newtheorem{definition}[theorem]{Definition}

\newtheorem{remark}[theorem]{Remark}

\begin{document}

\title{Low rank matrix recovery from rank one measurements}

\author{Richard Kueng}
\address{Institute for Physics, University of Freiburg, Rheinstra{\ss}e 10, 79104 Freiburg, Germany}
\curraddr{}
\email{richard.kueng@physik.uni-freiburg.de}
\thanks{}

\author{Holger Rauhut}
\address{Lehrstuhl C f{\"u}r Mathematik (Analysis), RWTH Aachen University, Pontdriesch 10, 52062 Aachen, Germany}
\curraddr{}
\email{rauhut@mathc.rwth-aachen.de}
\thanks{}

\author{Ulrich Terstiege}
\address{Lehrstuhl C f{\"u}r Mathematik (Analysis), RWTH Aachen University, Pontdriesch 10, 52062 Aachen, Germany}
\curraddr{}
\email{terstiege@mathc.rwth-aachen.de}
\thanks{}


\date{October 25, 2014}
\maketitle

\begin{abstract}
We study the recovery of Hermitian low rank 
matrices $X \in \mathbb{C}^{n \times n}$ 
from undersampled measurements via nuclear norm minimization. We consider
the particular scenario where the measurements are Frobenius inner products with random rank-one matrices of the form $a_j a_j^*$ for some
measurement vectors $a_1,\hdots,a_m$, i.e., the measurements are given by $y_j = \tr(X a_j a_j^*)$.
The case where the matrix $X=x x^*$ to be recovered is of rank one reduces to the problem of phaseless estimation 
(from measurements, $y_j = |\langle x,a_j\rangle|^2$ via the PhaseLift approach, 
which has been introduced recently.
We derive bounds for the number $m$ of measurements that guarantee successful uniform recovery of Hermitian rank $r$ matrices, 
either for the vectors $a_j$, $j=1,\hdots,m$, being chosen
independently at random according to a standard Gaussian distribution, or $a_j$ being sampled independently from an (approximate) complex projective
$t$-design with $t=4$. In the Gaussian case, we require $m \geq C r n$ measurements, while in the case of $4$-designs we need $m \geq Cr n \log(n)$. 
Our results are uniform in the sense that one random choice of the measurement vectors $a_j$ guarantees recovery of all rank $r$-matrices simultaneously with high probability. Moreover, we prove robustness of recovery under perturbation of the measurements by noise.
The result for approximate $4$-designs generalizes and improves a recent bound on phase retrieval due to Gross, Kueng and Krahmer.
In addition, it has applications in quantum state tomography. 
Our proofs employ the so-called bowling scheme which is based on recent ideas by Mendelson and Koltchinskii.
\end{abstract}

\section{Introduction}

\subsection{The phase retrieval problem}

The problem of retrieving a complex signal from measurements that are ignorant towards phases is abundant in many different areas of science,
such as X-ray cristallography \cite{ha93,mi90}, astronomy \cite{fi87} diffraction imaging \cite{chcoelsesh14,mi90} and more \cite{bacaed06-1,bu07,wa63}.
Mathematically formulated, the problem consists of recovering a complex signal (vector) $x \in \BC^n$ from measurements of the form
\begin{equation}
| \langle a_j, x \rangle |^2 = b_j \quad \textrm{for} \quad j=1,\ldots,m, \label{eq:measurements}
\end{equation}
where $a_1,\ldots,a_m \in \BC^n$ are sampling vectors. 
This ill-posed inverse problem is called \emph{phase retrieval} and has attracted considerable interest over the last few decades. 
An important feature of this problem is that the signal $x$ enters the measurement process (\ref{eq:measurements}) quadratically.
This leads to a non-linear inverse problem. Classical approaches to numerically solving it include alternating projection methods 
\cite{fi82-1,gesa72}. However, these methods usually require extra constraints and careful selection of parameters, and in particular, no rigorous
convergence or recovery guarantees seem to be available.

As Balan et al.\ pointed out in \cite{balan_painless_2009}, this apparent obstacle of having nonlinear measurements 
can be overcome by noting that the measurement process -- while quadratic in $x$ -- is linear in the outer product $x x^*$:
\begin{equation*}
| \langle a_j, x \rangle |^2 = \tr \left( a_j a_j^* x x^* \right).
\end{equation*}
This ``lifts'' the problem to a matrix space of dimension $n^2$, where it becomes linear and can be solved explicitly, provided that the number of measurements $m$ is at least $n^2$ \cite{balan_painless_2009}.
However, there is additional structure present, namely the  matrix $X = x x^*$ is guaranteed to have rank one. 
This connects the phase retrieval problem to the young but already extensive field of \emph{low-rank matrix recovery}. 
Indeed, it is just a special case of low-rank matrix recovery, where both the signal $X = x x^*$ and the measurement matrices $A_j = a_j a_j^*$ are constrained to be proportional to rank-one projectors.

It should be noted, however, that such a reduction to a low rank matrix recovery problem is just one possibility to retrieve phases. 
Other approaches use polarization identities \cite{alexeev_phase_2012} or alternate projections \cite{NJS13}. Yet another recent method is phase retrieval via Wirtinger flow \cite{caliso14}. 

\subsection{Low rank matrix recovery}
\label{sec:lowrankrec}
Building on ideas of compressive sensing \cite{cata06,do06-2,fora13}, low rank matrix recovery aims to reconstruct a matrix of low rank from incomplete
linear measurements via efficient algorithms \cite{recht_guaranteed_2010}. 
For our purposes we concentrate on Hermitian matrices $X \in \BC^{n \times n}$
and consider measurements of the form 
\begin{equation}
\tr \left( X A_j \right) = b_j \quad j=1,\ldots,m \label{eq:matrix_measurements}
\end{equation} 
where the $A_j \in \BC^{n \times n}$ are some Hermitian matrices.
For notational simplicity, we define the measurement operator
\begin{equation*}
\CA: \CH_n \to \BR^m \quad Z \mapsto \sum_{j=1}^m \tr \left( Z A_j \right) e_j,
\end{equation*}
where $e_1,\ldots,e_m$ denotes the standard basis in $\BR^m$.
This summarizes an entire (possibly noisy) measurement process via 
\begin{equation}
 b=\CA (X)+\epsilon. \label{eq:measurement_process}
\end{equation} Here $b = (b_1,\ldots,b_m)^T$ contains all measurement outcomes and $\epsilon \in \BR^m$ denotes additive noise. 
Low rank matrix recovery can be regarded as a non-commutative version of compressive sensing. 
Indeed, the structural assumption of low rank assures that the matrix is sparse in its eigenbasis. 
In parallel to the prominent role of $\ell_1$-norm minimization in compressive sensing \cite{fora13}, 
it is by now well-appreciated \cite{ahrero12,candes_exact_2009,capl11,recht_guaranteed_2010,gross_recovering_2011}
that in many relevant measurement scenarios, the sought for matrix $X$ can be efficiently recovered
via convex programming,
although the corresponding rank minimization problem is NP hard in general \cite{fa02-2}.

In order to formulate this convex program, we introduce the standard $\ell_p$-norm on $\BR^n$ or $\BC^n$ by
$\| x \|_{\ell_p} = (\sum_{\ell=1}^n |x_\ell |^p)^{1/p}$ for $1 \leq p < \infty$ and the Schatten-$p$-norm on the space $\CH_n$ of 
Hermitian $n\times n$ matrices as
\[
\|Z\|_p = \left( \sum_{\ell=1}^n \sigma_\ell(Z)^p \right)^{1/p} = \tr \left( |Z|^p \right)^{1/p}, \quad p\geq 1,
\]
where $\sigma_\ell(Z)$, $\ell=1,\hdots,n$, denote the singular values of $Z$, $\tr$ is the trace and $|Z| = (Z^*Z)^{1/2}$. Important special cases
are the nuclear norm $\|Z\|_* = \|Z\|_1$, the Frobenius norm $\|Z\|_F = \|Z\|_2$ and the spectral norm $\|Z\|_\infty = \|Z\|_{2 \to 2} = \sigma_{\max}(Z)$ being
the largest singular value. More information, concerning Schatten-$p$ norms can be found in Appendix~\ref{sec:Schatten}.
 
Assuming the upper bound $\| \epsilon \|_{\ell_2} \leq \eta$ on the noise for some $\eta \geq 0$, recovery via nuclear norm minimization corresponds to
\begin{equation}
\underset{Z \in \CH_n}{\textrm{minimize }}  \| Z \|_1  \
\textrm{subject to }  \| \CA (Z ) - b \|_{\ell_2} \leq \eta. \label{eq:convex_program}
\end{equation}
This is a convex optimization problem which can be solved computationally efficiently with various strategies \cite[Chapter 15]{fora13}, \cite{bova04,cope11,pabo13,toyu10}. We note that several alternatives to nuclear norm minimization may also be applied including
iteratively reweighted least squares \cite{forawa11}, iterative hard thresholding \cite{ceky14,tawe13}, greedy approaches \cite{brle10} and algorithms
specialized to certain measurement maps $\CA$ \cite{kemooh09}, but our analysis
is geared towards nuclear norm minimization and does not provide guarantees for these other algorithms.

Up to date, a number of measurement instances 
have been identified for which nuclear norm minimization (\ref{eq:convex_program}) -- and potentially other algorithms --
provably recovers the sought 
for low-rank matrix from considerably fewer than $n^2$ measurements 
\cite{candes_exact_2009,capl11,chparewi10,gross_recovering_2011,forawa11,liu_universal_2011,recht_guaranteed_2010,tr14}. 
All these constructions are based on randomness, the simplest being a random Gaussian measurement map where
all entries $\CA_{j,k,\ell}$ in the representation $\CA(X)_j = \sum_{k,\ell=1}^n \CA_{j,k,\ell} X_{k,\ell}$ are independent mean zero
variance one Gaussian random variables. It is shown in \cite{capl11,recht_guaranteed_2010} that 
\[
m \geq C r n
\]
measurements suffice in order to (stably) reconstruct a matrix $X \in \BC^{n \times n}$ of rank at most $r$ with probability at least $1-\exp(-c m)$, where
the constants $C,c > 0$ are universal.
This result is based on a version of the by-now classical restricted isometry property so that this result is uniform in the sense
that a random draw of $\CA$ enables reconstruction of {\em all} rank $r$ matrices simultaneously with high probability. A corresponding nonuniform
result, holding only for a fixed rank $r$ matrix $X$ is stated in \cite{chparewi10}, see also \cite{amlomctr13,tr14}, which shows that essentially $m > 6 r n$
measurements are sufficient, thus providing also good constants.

While unstructured Gaussian measurements provide optimal guarantees, which are comparably easy to derive, many applications demand for more structure in the measurement
process. A particular instance is the matrix completion problem \cite{bhchsawa13,candes_exact_2009,cata10-2,gross_recovering_2011,chen_incoherence_2013}, which aims at recovering missing entries
of a matrix which is known to be of low rank. Here, the source of randomness is in the selection of the known entries. In contrast to the unstructured measurements, additional
incoherence properties of the matrix to be recovered are required and the bounds on the number of measurements are slightly worse \cite{bhchsawa13,gross_recovering_2011}, 
namely $m \geq Cr n \log^2(n)$. The matrix completion setup generalizes to measurements with respect to an arbitrary operator basis. The incoherence assumption
on the matrix to be recovered can be dropped if in turn the operator basis is incoherent, which is the case for the particular example of Pauli measurements
arising in quantum tomography \cite{gross_recovering_2011,liu_universal_2011}. Here, a sufficient and necessary number of measurements scales like $m \geq Crn \log(n)$. 

Rank-one measurements, however, in general fail to be sufficiently incoherent for directly applying proof techniques of the same type.
For the particular case of phase retrieval (where the matrix of interest is by construction a rank-one projector) this obstacle could be overcome 
by providing problem specific recovery guarantees that either manifestly rely on (rank one) Gaussian measurements \cite{candes_solving_2012,tr14} 
or result in a non-optimal sampling rate \cite{grkrku13,candes_masked_2013,gross_improved_2014}.

\subsection{Weighted complex projective designs} 

The concept of real spherical designs was introduced  by Delsarte Goethals and Seidel in a seminal paper \cite{delsarte_spherical_1977} and has been studied in algebraic combinatorics \cite{sidelnikov_spherical_1999} and coding theory \cite{delsarte_spherical_1977,nebe_invariants_2001}.
Recently, complex projective designs -- the natural extension of real spherical designs to the complex unit sphere --
have been of considerable interest in quantum information theory 
\cite{zauner_quantendesigns_1999,scott_tight_2006,hayashi_reexamination_2005, gross_evenly_2007,low_large_2009,brandao_local_2012,lancien_distinguishing_2013}. 

Roughly speaking, a complex projective $t$-design is a finite subset of the complex unit sphere in $\BC^n$ with the particular property that the discrete average of any polynomial of degree $(t,t)$ (i.e., a polynomial $p(z,\bar{z})$ of total degree $t$ both in $z=(z_1,\hdots, z_n)$ and in $\bar{z}=(\bar z_1,\hdots,\bar z_n)$) or less equals its uniform average. 
Many equivalent definitions capture this essence, but the following one best serves our purpose.

\begin{definition}[\emph{exact, weighted $t$-design}, Definition 3 in \cite{scott_tight_2006}] \label{def:exact_design}
For $t \in \BN$, 
a finite set $\left\{w_1,\ldots,w_N \right\} \subset \BC^n$ of normalized vectors with corresponding weights $\left\{p_1,\ldots,p_N \right\}$ such that
$p_i \geq 0$ and $\sum_{i=1}^N p_i = 1$ is called a \emph{weighted complex projective $t$-design} of dimension $n$ and cardinality $N$  if 
\begin{equation}
\sum_{i=1}^N p_i \left( w_i w_i^* \right)^{\otimes t} = \int_{\BC P^{n-1}} \left( w w^* \right)^{\otimes t} \mathrm{d} w,
\label{eq:exact_designs}
\end{equation}
where the integral on the right hand side is taken with respect to the unique unitarily-invariant probability measure on the complex 
projective space $\BC P^{n-1}$ and the integrand is computed using arbitrary preimages of  the $w\in \BC P^{n-1}$ in the unit sphere in $\BC^n$. (Note that if $w_1$ and $w_2$ are elements of the unit sphere that have the same image $w$ in $ \BC P^{n-1}$ then $w_1w_1^*=w_2w_2^*$.) 
This definition in particular shows that uniform sampling from a $t$-design mimics the first $2t$ moments of sampling uniformly according to the Haar measure,
which is equivalent to sampling standard Gaussian vectors followed by renormalization.

A simple application of Schur's Lemma -- see e.g. \cite[Lemma 1]{scott_tight_2006} -- reveals that the integral on the right hand side of (\ref{eq:exact_designs}) amounts to
\begin{equation}
\int_{\BC P^{n-1}} \left( w w^* \right)^{\otimes t} \mathrm{d}w = \binom{n+t-1}{t}^{-1} P_{\Sym^t}, \label{eq:PSym_Haar}
\end{equation}
where $P_{\Sym^t}$ denotes the projector onto the totally symmetric subspace $\Sym^t$ of $\left( \BC^n \right)^{\otimes t}$ defined in the appendix --
see equation (\ref{eq:symmetrizer}). 
\end{definition}
{In accordance with \cite{matthews_distinguishability_2009}, we call a $t$-design \emph{proper}, if all the weights are equal, i.e., $p_i = 1/N$ for all $i=1,\ldots,N$.

Although exact, proper $t$-designs exist and can be constructed in any dimension $n$ for any $t \in \BN$ \cite{seymour_averaging_1984,bajnok1992construction,korevaar_chebyshev_1994,hayashi_reexamination_2005}, these constructions are typically inefficient in the sense that they require vector sets of exponential size.
For example, the construction in \cite{hayashi_reexamination_2005} requires on the order of $\CO \left( t \right)^n$ vectors which scales exponentially 
in the dimension $n$. Constructions of \emph{exact, proper} designs with significantly smaller number of vectors (scaling only polynomially in $n$)
are notoriously difficult to find.

By introducing weights, it becomes simpler to obtain designs with a number of elements that scales polynomially in the dimension $n$.
Some existence results can be found in \cite{hapa05}, where weighted $t$-designs appear under the notion of cubatures of strength $t$.
It seems that one can construct weighted $t$-designs by drawing sufficiently many vectors at random and afterwards solving a linear system for
the weights. Further note, that generalizations of cubatures to higher dimensional projections were used in \cite{bael12} in the context of a generalized phase retrieval
problem, where the measurements are given as norms of projections onto higher dimensional subspaces.

\section{Main results}

\subsection{Low rank matrix recovery from rank one Gaussian projections}

Our first main result gives a uniform and stable guarantee for recovering rank-$r$ matrices with $\mathcal{O} (r n)$ rank one measurements that are proportional to projectors onto standard Gaussian random vectors.

\begin{theorem} \label{Th1}
Consider the measurement process described in (\ref{eq:measurement_process}) with measurement matrices $A_j = a_j a_j^*$, where $a_1,\ldots,a_m \in \BC^n$ are  independent standard Gaussian distributed random vectors. 
Furthermore assume that the number of measurements $m$ obeys
\begin{equation*}
m \geq  C_1 n r,
\end{equation*}
for $1 \leq r \leq n$ arbitrary. 
Then with probability at least $1 - \mathrm{e}^{-C_2 m}$ it holds that for any positive semidefinite matrix $X \in \CH_n$ with rank at most $r$, any solution $X^\#$  to the convex optimization problem (\ref{eq:convex_program}) with noisy measurements $b = \CA (X)+\epsilon$, where $\| \epsilon \|_{\ell_2} \leq \eta$,
obeys
\begin{equation}\label{err:bound1}
\| X - X^\# \|_2 \leq \frac{C_3 \eta}{\sqrt{m}}.
\end{equation}
 Here, $C_1,C_2$ and $C_3$ denote universal positive constants. (In particular, for $\eta=0$ one has exact reconstruction.)
\end{theorem}

For the rank one case $r=1$, Theorem \ref{Th1} essentialy reproduces the main result in \cite{candes_solving_2012} which uses completely different proof techniques. (More precisely,  for $X$ of rank $1$ the estimate in loc. cit. is $\| X - X^\# \|_2 \leq \frac{C \Vert \epsilon \Vert_1}{{m}}$ with high probability.)
A variant of the above statement was  shown in \cite{tr14} to hold (in the real case) for a fixed matrix $X$ of rank one. (More precisely, in loc. cit. it is assumed that $X$ is positive semidefinite and the optimization  is performed wrt. the function $f$ given by (\ref{f:posdef}) below.)
 In fact, our proof reorganizes and extends the arguments of \cite[Section~8]{tr14} in such a way, that Theorem~8.1 of loc.\ cit.\ is shown to hold even uniformly (that is simultaneously for all $X$) and for arbitrary rank.
On the contrary to \cite{candes_solving_2012}, we will not need $\varepsilon$-nets to show uniformity.

\subsection{Recovery with 4-designs} 

As we will see, the proof method for Theorem~\ref{Th1} can also be applied to measurements drawn independently 
from a weighted complex projective $4$-design in the sense of Definition~\ref{def:exact_design}. 
In \cite{grkrku13} exact complex projective $t$-designs have been applied to the problem of phase retrieval. 
The main result (Theorem 1) in \cite{grkrku13} is a non-uniform exact recovery guarantee for phase retrieval via the convex optimization problem (\ref{eq:convex_program}) 
that requires $m = \mathcal{O} \left( t n^{1+2/t} \log^2 n \right)$ measurement vectors that are drawn uniformly from a proper $t$-design ($t \geq 3$). 
The proof technique which we are going to employ here, allows for considerably generalizing and improving this statement. 
We will draw the measurement vectors $a_1,\hdots,a_m \in \BC^n$ independently at random
from a weighted $4$-design $\left\{p_i, w_i \right\}_{i=1}^N$, which means that for each draw of $a_j$, the design element
$w_i$ is selected with probability $p_i$. In the sequel we assume that $n\geq 2$.

\begin{theorem} \label{Th2}
Let $\left\{p_i,w_i \right\}_{i=1}^N$ be a weighted $4$-design and
consider the measurement process described in (\ref{eq:measurement_process}) with measurement matrices $A_j = \sqrt{n(n+1)}a_j a_j^*$, where $a_1,\ldots,a_m \in \BC^n$ are drawn independently 
from $\left\{p_i, w_i \right\}_{i=1}^N$. 
Furthermore assume that the number of measurements $m$ obeys
\begin{equation*}
m \geq  C_4 n r \log n,
\end{equation*}
for $1 \leq r \leq n$ arbitrary. 
Then with probability at least $1 - \mathrm{e}^{-C_5 m}$ it holds that for any  $X \in \CH_n$ with rank at most $r$, 
any solution $X^\#$  to the convex optimization problem (\ref{eq:convex_program})  
with noisy measurements $b = \CA (X)+\epsilon$, where $\| \epsilon \|_{\ell_2} \leq \eta$, obeys
\begin{equation}\label{err:bound2}
\| X - X^\# \|_2 \leq \frac{C_6 \eta}{\sqrt{m}}.
\end{equation}
 Here, $C_4,C_5,C_6>0$ again denote universal positive constants. 
\end{theorem}

The normalization factor $\sqrt{n(n+1)}$ leads to approximately the same normalization of the $A_j$ (wrt. the Frobenius norm) as in expectation in the Gauss case.
The theorem is a stable, uniform guarantee for recovering arbitrary Hermitian matrices of rank at most $r$ 
with high probability using the convex optimization problem (\ref{eq:convex_program}) 
and $m = \mathcal{O} \left( nr \log (n) \right)$ measurements drawn independently (according to the design's weights) from a weighted 4-design. 
It obviously covers sampling from \emph{proper} 4-designs as a special case. 

Also, Theorem \ref{Th2} is close to optimal in terms of the design order $t$ required. In the context of the phase retrieval 
problem\footnote{i.e., recovering unknown Hermitian matrices of rank one} it was shown in  \cite[Theorem 2]{grkrku13}, 
that choosing measurements uniformly from a proper 2-design does not allow for a sub-quadratic sampling rate $m$ without 
additional structural assumptions on the measurement ensemble. 
It is presently open whether Theorem~\ref{Th2} also holds for $3$-designs.

Finally, note that the results for Gaussian measurement vectors and 4-designs are remarkably similar. 
They only differ by a logarithmic factor. 
This underlines the usefulness of complex projective designs as a general-purpose tool 
for de-randomization -- see e.g. \cite[Section 1.1.]{grkrku13} for further reading on this topic.
Also, Theorem~\ref{Th2} resembles insights in the context of distinguishing quantum states
\cite{matthews_distinguishability_2009,ambainis_quantum_2007},
where it was pointed out that (approximate) $4$-designs ``perform almost as good'' as uniform measurements (projectors onto random Gaussian vectors).
Note that we will generalize Theorem~\ref{Th2} to approximate $4$-designs in Theorem~\ref{Th3} below.

\subsection{Extensions}

In this section we state variants of the main theorems which can be proved in a similar way.

\subsubsection{Real-valued case}
\label{subsub:real_Gauss}
Theorem~\ref{Th1} is also valid in the real case, i.e., assuming that the $a_j$ are real standard Gaussian distributed and  $\CH_n$ is replaced by the space $\CS_n$ of real symmetric $n\times n$-matrices. 
The proof of the corresponding statement is very similar to the one of Theorem~\ref{Th1} and we sketch the necessary adaptations in Subsection~\ref{sub:real_Gauss_proof}.

\subsubsection{Recovery of positive semidefinite matrices} \label{subsub:psd}
The matrix $X$ to be recovered 
may be known to be positive semidefinite ($X \succcurlyeq 0$) 
in advance. In this case, one can enforce the reconstructed matrix to be positive semidefinite by
considering the optimization program
\[
\underset{Z \succcurlyeq 0}{\textrm{minimize }}  \tr(Z)  \quad 
\textrm{ subject to } \quad  \| \CA (Z ) - b \|_{\ell_2} \leq \eta 
\]
instead of the nuclear norm minimization program \eqref{eq:convex_program}. 
Then analog versions of Theorems~\ref{Th1}, \ref{Th2} and \ref{Th3} hold. In particular, the error bounds \eqref{err:bound1}, \eqref{err:bound2} remain valid. 
In the noisy case $\eta > 0$, this does not follow directly from these theorems, since the minimizer of the nuclear norm minimization
\eqref{eq:convex_program} is not guaranteed to be positive semidefinite in the noisy case. The proof proceeds similarly as the ones for the case $X \in \CH_n$. Instead of
the nuclear norm one has to consider (as in  \cite{tr14}) the function
\begin{equation}\label{f:posdef}
f:\CH_n\to \BR\cup\{\infty \}, 
\quad f(X)=\begin{cases}
  \tr(X),  & \text{if } X\succcurlyeq 0 \\
   \infty, & \text{otherwise. } \end{cases}
\end{equation}

\section{Applications to quantum state tomography}

A particular instance of matrix recovery is the task of reconstructing a finite $n$-dimensional quantum mechanical system
which is fully characterized by its \emph{density operator} $\rho$ -- an $n \times n$-dimensional positive semidefinite matrix with trace one. 
Estimating  the density operator of an actual (finite dimensional) quantum system is an important task in quantum physics known as \emph{quantum state tomography}. 

One is often interested in performing tomography for quantum systems  that have certain structural properties. 
An important structural property -- on which we shall focus here -- is \emph{purity}. 
A quantum system is called \emph{pure}, if its density operator has rank one and \emph{almost pure} if it is well approximated by a matrix of low rank $\mathrm{rank}(\rho)=r \ll n$. 
Assuming this structural property, quantum state tomography is a low-rank matrix recovery problem \cite{gross_quantum_2010,gross_recovering_2011,flammia_quantum_2012,liu_universal_2011}. 
An additional requirement for tomography is the fact that the measurement process has to be ``experimentally realizable'' and -- preferably -- ``efficiently'' so.

Any ``experimentally realizable'' quantum mechanical measurement corresponds to a \emph{positive operator-valued measure} (POVM). 
In the special case of (finite) $n$-dimensional quantum systems, a POVM is a set of positive semidefinite matrices
$\left\{M_j\right\}_{j=1}^N \subset \CH_n$ that sum up to the identity, i.e., 
$\sum_{j=1}^N M_j = \id$ -- see e.g. \cite[Chapter 2.2.6]{nielsen_quantum_2010} for further information. 

For practical reasons, it is highly desirable that a quantum measurement (represented by a POVM) 
can be implemented with reasonable effort.
In accordance with \cite{nielsen_quantum_2010}, we call a POVM-measurement \emph{efficient} (or \emph{practical}), 
if it can be carried out by performing a number of $\mathcal{O} \left( \polylog (n) \right)$ elementary steps\footnote{This notion is comparable to the \emph{circuit depth} in classical computer science.}. 
Making this notion precise would go beyond the scope of this work and we refer to \cite{ambainis_quantum_2007,nielsen_quantum_2010} for further reading. 

Below we will concentrate on random constructions of the vectors $a_j$.
We note, however, that implementing the POVM element $a_j a_j^*$ corresponding to the projection onto a Gaussian random vector
is {\em not} efficient 
as it requires $\CO \left( \mathrm{poly}(n) \right)$ steps. 
This renders all low rank matrix recovery guarantees which rely on Gaussian measurements -- like in Theorem \ref{Th1} above
-- inefficient (and therefore impractical) for low rank quantum state tomography.
Utilizing a weakened concept of $t$-designs discussed next, we partly overcome this obstackle with
Theorem~\ref{Th3} below and its possible implementations outlined in Sections~\ref{subsub:ambainis_POVM}, \ref{subsub:unitary}.

\subsection{An analogue of Theorem \ref{Th2} for approximate designs} \label{sub:approx_designs}

While Theorem \ref{Th2} is a substantial derandomization of Theorem \ref{Th1} and therefore interesting from a theoretical point of view, 
its usefulness hinges on the availability of constructions of exact weighted 4-designs. Unfortunately, such constructions are notoriously difficult to find unless
one relies on randomness, for which, however, the resulting designs are {\em not} efficient in the sense described in the previous section.
One way to circumvent these difficulties 
is to relax the defining property (\ref{eq:exact_designs}) of a $t$-design.
This approach was -- up to our knowledge -- introduced by A.~Ambainis and J.~Emerson \cite{ambainis_quantum_2007}
and resulted in the notion of approximate designs which is by now well established in quantum information science. 

\begin{definition}[\emph{Approximate $t$-design}] \label{def:approx_designs}
We call a weighted set  $\left\{p_i,w_i \right\}_{i=1}^N$ of normalized vectors
an approximate $t$-design of $p$-norm accuracy $\theta_p$, if
\begin{equation}
 \left\| \sum_{i=1}^N p_i \left( w_i w_i^* \right)^{\otimes t} - \int_{\BC P^{n-1}} \left( w w^* \right)^{\otimes t} \mathrm{d}w \right\|_p \leq  \binom{n+t-1}{t}^{-1} \theta_p.
\label{eq:approx_designs}
\end{equation}
\end{definition}

While accuracy measured in arbitrary Schatten-$p$-norms is conceivable, the ones measured in  operator norm ($p = \infty$) \cite{hayden_randomizing_2004,ambainis_quantum_2007,low_pseudo_2010,brandao_local_2012} and nuclear norm ($p =1$) \cite{nakata_generating_2014} are the ones most commonly used -- at least in quantum information theory.
For these two accuracies, 
the definition in particular assures that every approximate $t$-design is in particular also a $k$-design for any $1 \leq k \leq t$ with the same $p$-norm accuracy $\theta_p$ \cite{ambainis_quantum_2007,low_pseudo_2010}. 
For the sake of being self-contained we provide a proof of this statement in the appendix -- see Lemma \ref{lem:design_hierarchy}.

A slightly refined analysis reveals that Theorem~\ref{Th2} also holds for sufficiently accurate approximate 4-designs.

\begin{theorem} \label{Th3}
Fix $1 \leq r \leq n$ arbitrary and let $\left\{p_i, w_i \right\}_{i=1}^N$ be an approximate 4-design satisfying
\begin{eqnarray}
\left\| \sum_{i=1}^N p_i w_i w_i^* - \frac{1}{n} \id \right\|_\infty &\leq& \frac{1}{n}, \label{eq:approx_tight_frame}
\end{eqnarray}
that admits either operator norm accuracy {$\theta_\infty \leq 1/(16r^2)$}, or trace-norm accuracy $\theta_1 \leq 1/4$, respectively. 
Then, the recovery guarantee from Theorem \ref{Th2} is still valid (possibly with slightly worse absolute constants $\tilde{C}_4, \tilde{C}_5$ and $\tilde{C}_6$).  
\end{theorem}

\subsection{Protocols for efficient low rank matrix recovery} \label{sub:protocols}

Up to now, efficient recovery of low rank density operators by means of the convex optimization problem (\ref{eq:convex_program}) has been established 
for random measurements of (generalized) Pauli observables \cite{gross_quantum_2010,gross_recovering_2011}. For this type of measurements, the statistical issues are well understood \cite{flammia_quantum_2012} and Y.K.~Liu managed to prove a uniform recovery guarantee \cite{liu_universal_2011} which is comparable to the results presented here. Also, this procedure has been tested in experiments \cite{schwemmer_experimental_2014}. 

Theorem \ref{Th3} is similar in spirit and we show here that it permits efficient low rank quantum state tomography for different types of measurements. 
Indeed, in the field of quantum information theory,
 various ways of constructing approximate $t$-designs are known. 
Most of these methods are inspired by ``realistic'' quantum mechanical setups (e.g. the circuit model \cite[Chapter 4]{nielsen_quantum_2010})
and can therefore be -- in principle -- implemented efficiently in an actual experiment. 

Introducing these constructions in full detail would go beyond the scope of this work and we content ourselves with sketching two possible ways of generating
approximate 4-design measurements which meet the requirements of Theorem \ref{Th3}. 
For further clarification on the concepts used here, we refer directly to the stated references.

From now on we shall assume that the dimension $n = 2^d$ is a power of two ($d$-qubit density operators).

\subsubsection{The Ambainis-Emerson POVM} \label{subsub:ambainis_POVM}
In \cite{ambainis_quantum_2007}, the authors provide a way of constructing a normalized  approximate $4$-design of operator-norm accuracy
$\theta_\infty = \mathcal{O} \left( 1/ n^{1/3}\right)$, which in addition is a tight frame. 
They furthermore present a way to generate the corresponding POVM-measurements efficiently -- i.e., involving only $\mathcal{O} \left( \polylog(n) \right)$ elementary steps. 
It therefore meets the requirements of Theorem~\ref{Th3}, provided that the maximal rank $r$ of the unknown density operator obeys
\begin{equation}
r \leq C_7 n^{1/6}, \label{eq:rank_constraint}
\end{equation}
where $C_7$ is a sufficiently small absolute constant. 
The additional rank requirement stems from the fact that the resulting design only has limited accuracy.

This accuracy can be improved if we construct an approximate design  in a much larger space -- say $\BC^{n^6}$ -- and project it down onto an arbitrary $n$-dimensional subspace.
The reason for such an approach is that the projected design's accuracy corresponds to
$\theta_\infty = 
\CO \left( \left( n^6 \right)^{-1/3} \right) = \mathcal{O}(1/n^2)$. 
This allows for replacing (\ref{eq:rank_constraint}) by the much weaker rank constraint
\begin{equation}
r \leq C_8 n, \label{eq:weaker_rank_constraint}
\end{equation}
(where $C_8$ is again a sufficiently small absolute constant) in order to assure that the design's operator-norm accuracy obeys $\theta_\infty \leq 1/(16 r^2)$.

Also, the projected design vectors still form a tight frame, but are sub-normalized, i.e. $\| \tilde{w}_i \|_{\ell_2}^2 = \| P w_i \|_{\ell_2}^2 \leq \| w_i \|_2^2 = 1$.
Here, $P:\BC^{n^6} \to \BC^n$ denotes the projection.  
However, since they are an approximate design's projection onto a smaller space, they maintain all properties of an approximate $4$-design -- most notably Lemma \ref{lem:design_hierarchy} -- 
except normalization. In the proof of Theorem \ref{Th3}, normalization is only used once, namely in (\ref{eq:matrix_Chernoff1}) and sub-normalization is sufficient to guarantee 
this estimate. 
Consequently, Theorem \ref{Th3} is applicable and guarantees universal quantum state tomography via the convex optimization problem (\ref{eq:convex_program}), 
provided that (\ref{eq:weaker_rank_constraint}) holds and $m = C_4 r n \log n$ randomly chosen measurements $\tr \left( \tilde{w}_i \tilde{w}_i^* \rho \right)$ are known.

\subsubsection{Approximate unitary designs} \label{subsub:unitary}
Another way to generate approximate $t$-designs is to consider arbitrary orbits of unitary $t$-designs.
Unitary $t$-designs $\left\{p_i, U_i \right\}_{i=1}^N$ are a natural generalization of the spherical design concept to unitary matrices \cite{dankert_exact_2006,gross_evenly_2007}.
They have the particular property that every weighted orbit $\left\{p_i, U_i x \right\}$ with $\|x\|_{\ell_2} = 1$
of an approximate unitary design forms an approximate complex projective $t$-design of the same accuracy.

It was shown in \cite{brandao_local_2012} that unitary $t$-designs of arbitrary operator-norm accuracy $\theta_\infty$ 
can be constructed efficiently by using local random circuits. 
This approach allows for generating an approximate unitary 4-design of operator-norm accuracy $\theta_\infty \leq 1/(16 n^2)$ by
means of local random circuits of length $C_9 \log (n)^2$, where $C_9$ is a sufficiently large absolute constant. 
Consequently, every orbit of the union of all such local random circuits of length $C_9 \log (n)^2$ forms a normalized approximate 4-design 
which meets the requirements of Theorem \ref{Th3}.
One way of implementing such a measurement consists in choosing a local quantum circuit $U_i$ at random, applying its adjoint circuit $U_i^*$ to the density operator $\rho$ and 
then measuring the two-outcome POVM $\left\{ x x^*, \id - x x^* \right\}$, where $x \in \BC^n$ is arbitrary (but fixed and normalized) to obtain
\begin{equation*}
y_i = \tr \left( x x^* U_i^* \rho U_i \right) = \tr \left(U_i x x^* U_i^* \rho \right) = \tr \left( w_i w_i^* \rho \right).  
\end{equation*}
According to Theorem \ref{Th3}, $m = \tilde{C}_4 nr \log n$ random measurements of this kind are sufficient to reconstruct any density operator $\rho$ of rank at most $r$ with very high probability via the convex optimization problem (\ref{eq:convex_program}).

\begin{remark}
One should note that the approximate unitary designs of \cite{brandao_local_2012} 
are not of a finite nature, because the set of all local random unitaries is continuous. Nevertheless, assuming that such local random unitaries are available as ``basic building blocks'', local random circuits are efficiently implementable in terms of circuit length.
Replacing the atomic expectation values $\sum_{i=1}^N p_i \left( w_i w_i \right)^{\otimes t}$ by 
their continuous counterparts does not change the argument and Theorem~\ref{Th3} remains valid. 
\end{remark}

It is worthwhile to point out that the two possible applications of Theorem \ref{Th3} to the problem of low rank quantum state tomography, as presented here, are not yet optimal.
The implementation using the Ambainis-Emerson POVM -- presented in \ref{subsub:ambainis_POVM} -- suffers from the drawback that it demands either a very strong criterion on the density operator's rank -- condition (\ref{eq:rank_constraint}) -- or generating the design in a much larger space and projecting it down. The latter construction is highly unlikely to be optimal and it is furthermore a priori not clear where the corresponding POVM-measurements can be implemented efficiently.

The second approach, on the other hand, suffers from the drawback that carrying out each of the $C r n \log n$ random measurements requires terminating with a very coarse two-outcome POVM measurement. 
It is very likely that a more fine grained-output statistics could be obtained with comparable effort. 
The recovery protocol stated here, however, does not allow for advantageously taking into account such refined information about the unknown state. 

However, we still feel that mentioning these protocols is worthwhile, as they substantially narrow down the gap between what can be proved (Theorem \ref{Th3} and the protocols presented in subsection \ref{sub:protocols}) and what can be implemented efficiently in an actual quantum state tomography experiment. 
Next, we provide ideas for further narrowing this gap 
and finding more protocols that allow for efficient low rank quantum state tomography.

\subsection{Outlook}

The construction of approximate $t$-designs in Section~\ref{subsub:ambainis_POVM} via projections from higher-dimensional designs
would be much stronger if an efficient protocol for the corresponding POVM measurements could be provided. We leave this for future work.
Alternatively, the authors of \cite{ambainis_quantum_2007} mention results by Kuperberg \cite{ku06} who managed to construct exact $t$-designs 
containing only $\CO \left( n^{2t} \right)$ vectors. They furthermore conjecture that their method of efficiently implementing the corresponding POVM measurement also works for Kuperberg's exact construction. 
Trying to find such an implementation and combining it with Theorem~\ref{Th2} also does constitute an intriguing follow up-project.

{\em Diagonal-unitary designs} are yet another generalization of the spherical design concept to a more restrictive family of unitaries \cite{nakata_generating_2014}.
The notion of a diagonal-unitary design 
depends on choosing a reference basis and is therefore weaker than the unitary design notation from above.
Nevertheless, in \cite[Proposition 1]{nakata_generating_2014} it was shown that the orbit\footnote{
For a diagonal-unitary design with respect to the standard basis $e_1,\ldots,e_n$, their result requires the first Fourier vector $f_1 = \frac{1}{\sqrt{n}} \sum_{i=1}^n e_i$ 
as a fiducial.
This vector is isomorphic to the $| + \rangle^{\otimes d} = \left( \frac{1}{\sqrt{2}} \left( e_1 + e_2 \right) \right)^{\otimes d}$ state which is well-known in quantum information theory.
}
of a particular vector $f_1 \in \BC^n$ under a diagonal-unitary $t$-designs still forms approximate complex projective $t$-designs
with trace-norm accuracy 
\begin{equation}
\theta_1 = \binom{n+t-1}{t} \left( \frac{t(t-1)}{n} + \CO \left( \frac{1}{n^2} \right) \right). \label{eq:diagonal_unitary_design}
\end{equation}
A quick calculation reveals that this orbit forms a normalized tight frame. 
Unfortunately, the trace-norm accuracy (\ref{eq:diagonal_unitary_design}) is too weak for a direct application of Theorem \ref{Th3}. 
However, in \cite[Theorem 1]{nakata_generating_2014} it is shown that the union of all $3$-qubit phase-random circuits forms an exact diagonal-unitary 4-design.
Similar to local random circuits, such $3$-qubit phase-random circuits can in principle be implemented efficiently \cite[Proposition 3]{nakata_generating_2014} in an actual quantum mechanical setup.
Furthermore, comparing (\ref{eq:diagonal_unitary_design})  with the accuracy relation $\theta_\infty \leq \theta_1 \leq n^t \theta_\infty$ -- see Lemma \ref{lem:design_hierarchy} in the appendix --  suggests that particular orbits of diagonal-unitary designs might possess a much tighter operator-norm accuracy, if the spectrum of their ($t$-fold tensored) average were sufficiently flat. Such a result, combined with Theorem \ref{Th3}, would lead to a tomography procedure that is similar to the one of Section~\ref{subsub:unitary}, but uses random 3-qubit phase gates instead of local random circuits.

\section{Proofs}

Our proof technique consists in the application of a uniform version of Tropp's bowling scheme, see \cite{tr14}.
The crucial ingredient is a new method due to Mendelson \cite{me14} and Koltchiskii, Mendelson \cite{kome13} (see also \cite{leme14}) 
to obtain lower bounds for quantities of the form $\inf_{u \in E} \sum_{j=1}^m |\langle \phi_j,x \rangle|^2$ 
where the $\phi_j$ are independent random vectors in $\BR^d$ and $E$ is a subset of $\BR^d$.
We start by recalling from \cite{tr14} the notions and results underlying this technique.

Suppose we measure $x_0\in \BR^d$ via measurements $y=\Phi x_0 + \epsilon\in \BR^m$, where $\Phi$ is an  $m\times d$ measurement matrix and $\epsilon \in \BR^m$ vector of unknown errors. Let $\eta\geq 0$ and assume   $\Vert \epsilon\Vert_{\ell_2}\leq \eta$.
For
$f:\BR^d \to \BR\cup\{\infty \} $  proper convex  we aim at recovering $x_0$ by solving the
 convex program  
\begin{equation}\label{eq:conv:prog:f}
\text{minimize } f(x) \quad \text{ subject to } \quad \Vert \Phi x-y\Vert_{\ell_2}\leq \eta.
\end{equation}
Here, {\em proper convex} means that $f$ is convex and attains at least one finite value.

\noindent Let $K\subseteq \BR^d$ be a cone. Then we define the minimum singular value of $\Phi$ with respect to $K$ as 
$$
\lambda_{\min}(\Phi; K)=\inf\{\Vert \Phi u\Vert_{\ell_2}:  u \in K\cap \BS^{d-1}\}, 
$$
where $\BS^{d-1}$ is the unit sphere in $\BR^d$.
For $x\in \BR^d$, we consider the (convex)  {\em descent cone} 
$$
\CD(f,x)=\bigcup_{\tau> 0} \{y\in \BR^d:\  f(x+\tau y)\leq f(x)\}.
$$
With these notions, the success of the convex program \eqref{eq:conv:prog:f} can be estimated as follows.
\begin{proposition} \label{CRPWT}(\cite{tr14}, see also \cite{chparewi10})
Let $x_0\in \BR^d$, $\Phi\in \BR^{m\times d}$ and  $y=\Phi x_0 +  \epsilon$ with $ \Vert  \epsilon \Vert_{\ell_2}\leq \eta$. Let $f:\BR^d \to \BR\cup\{\infty \} $  be proper convex and let $x^{\sharp}$  be a solution of the corresponding convex program \eqref{eq:conv:prog:f}. Then 
$$
\Vert  x^{\sharp}-x_0\Vert_{\ell_2}\leq \frac{2\eta}{\lambda_{\min}(\Phi; \CD(f,x_0))}.
$$
\end{proposition}
\vspace{0.2cm}

The crucial point for us is that in  the situation that $\Phi$ is a random matrix with i.i.d. rows, 
the following theorem can be applied to estimate $\lambda_{\min}(\Phi; \CD(f,x_0))$ (see also \cite{kome13, tr14, me14}).

\begin{theorem}\label{KMT} (Koltchinskii, Mendelson; Tropp's version \cite{tr14})
Fix $E\subset \BR^d$ and let $\phi_1,\hdots,\phi_m$ be independent copies of a random vector $\phi$ in $\BR^d$. 
For $\xi >0$ let 
\begin{align*}
Q_{\xi}(E;\phi) & =\inf_{u\in E}\BP\{ \vert \langle \phi, u \rangle \vert \geq \xi\}\\
\mbox{ and } \quad W_m(E,\phi) &=\BE \sup_{u\in E}\langle h,u\rangle, \quad \text{ where } h=\frac{1}{\sqrt m} \sum_{j=1}^{m}\varepsilon_j \phi_j
\end{align*}
with $(\varepsilon_j)$ being a Rademacher sequence\footnote{A Rademacher vector $\epsilon = (\epsilon_j)_{j=1}^m$ is a vector of independent Rademacher random variables, taking the values $\pm 1$ with equal probability.}. Then for any $\xi >0$ and any $t\geq 0$ with probability at least $1-e^{-2 t^2}$ 
$$
\inf_{u\in E} \left( \sum_{i=1}^{m}\vert \langle \phi_i, u \rangle \vert^2  \right )^{1/2}\geq \xi \sqrt m Q_{2\xi}(E;\phi)-2W_m(E,\phi)-\xi t.
$$
\end{theorem}
\begin{remark}
We note that the above theorem is stated in \cite{tr14} to hold with probability $1-e^{-t^2/2}$. Inspecting the proof, however, reveals
that the probability estimate can actually be improved to $1-e^{-2t^2}$.
\end{remark}

We will apply the notions in these results in the context of Theorems \ref{Th1} and \ref{Th2} as follows:
   
\begin{itemize}
\item identify $\CH_n$ with $\BR^{d}=\BR^{n^2}$
\item $\Phi$ is the matrix of $\CA$ in the standard basis, i.e., $\Phi (X)_i=\tr(a_ia_i^*X)$ \ \ 
\item $f:\CH_n\to \BR\cup\{\infty \}$ is the nuclear norm, i.e., $f(X)=\Vert X \Vert_1.$ 
\end{itemize}
In particular, $$\CD(f,X)=\bigcup_{\tau> 0} \{Y\in \CH_n:\  f(X+\tau Y)\leq f(X)\}.$$

In Topp's original \emph{bowling scheme}, \cite[Sections 7 and 8]{tr14}, a positive semidefinite matrix $X$ of rank $1$ 
is fixed and Theorem \ref{KMT} is then applied to $E_X=\CD(f,X)\cap \BS^{d-1}$, where $\BS^{d-1} = \left\{ Z \in \CH_n:\; \| Z \|_2 = 1 \right\}$. He then uses the Payley-Zygmund inequality to obtain a lower bound for $Q_{2 \xi}$ (after choosing some appropriate $\xi$)
and finally applies arguments like conic duality to bound $W_m$ from above. 

Our approach  differs from the original bowling scheme in one aspect: instead of fixing one rank $r$-matrix and focusing on $E_X$, we are going to consider the union 
$E_r =\left\{ X \in \CH_n: \; \rank (X) \leq r, \; X \neq 0 \right\}$ of all low rank matrices. 
The rest of the proof essentially parallels the bowling scheme from \cite{tr14}. However, we are going to require an auxiliary statement -- Lemma \ref{absch} below -- in order to obtain a comparable upper bound on $W_m$. This slightly refined analysis is going to result in a uniform recovery result whose probability of success equals the one for non-uniform recovery of a single fixed $X$.
Note that with such an approach, we do not need to use $\varepsilon$-nets in order to establish uniformity.

 For  $r\leq n$ let 
$$
K_r=\bigcup_X \CD(f,X), 
$$ where the union runs over all $X\in \CH_n\setminus \{0\}$ 
of rank at most $r$.
We further define $$E_r=K_r\cap \BS^{d-1}=\bigcup_X E_X,$$ where $E_X=  \CD(f,X)\cap \BS^{d-1}$.
We recall 
that for a convex cone $K\subseteq \BR^d$, its {\em polar cone} is defined to be the closed convex cone 
$$
K^{\circ}=\{ v\in \BR^d:  \langle v,x \rangle \leq 0 \text{ for all } x\in K  \}.
$$

\noindent 
A crucial ingredient for Theorems~\ref{Th1} and \ref{Th2} is the following lemma.

\begin{lemma}\label{absch}
Let $A\in\CH_n$ be a Hermitian $n\times n$-matrix. Then $$
\sup_{Y\in E_r} \tr(A\cdot Y)\leq 2\sqrt{r} \Vert A\Vert_{\infty}.
$$
\end{lemma}
By duality and the matrix H{\"o}lder inequality this statement is equivalent to
\begin{equation}
\| Y \|_1 \leq 2 \sqrt{r} \quad \mbox{ for all } Y \in E_r. \label{eq:absch1}
\end{equation}
The following proof is inspired by  \cite[Section 8]{tr14}, where similar arguments are used.
\begin{proof}
It is enough to show that, for any  $X\in\CH_n\setminus \{0\}$   of rank at most $r$, we have  
$$
\sup_{Y\in E_X} \tr(A\cdot Y)\leq 2\sqrt{r} \Vert A\Vert_{\infty}.
$$ 
We  may  assume that $X$ has precisely rank $r\geq 1$.
 By weak duality for cones, see \cite[Proposition~4.2]{tr14} or \cite[eq.\ (B.40)]{fora13}, we have $\sup_{Y\in E_X} \tr(A\cdot Y)\leq \dist_F(A, \CD(f,X)^{\circ})$, where as usual $ \dist_F(A, \CD(f,X)^{\circ})=\inf_{B\in \CD(f,X)^{\circ}}\Vert A-B\Vert_2$. 
By \cite[Fact 4.3]{tr14}, we know that the polar cone $\CD(f,X)^{\circ}$ is the closure of $\bigcup_{\tau\geq 0}\tau\cdot \partial  f(X)$. 
For $S \in \partial f(X)$ and $\tau\geq 0$, it follows that 
$$
\sup_{Y\in E_X} \tr(A\cdot Y)\leq \Vert A-\tau\cdot S\Vert_2.
$$
Write $X=\sum_{i=1}^{r}\lambda_i x_i x_i^*$, where the $x_i$ are orthonormal and the $\lambda_i$ are non-zero.  
Extend $x_1,\hdots,x_r$ to an orthonormal basis $x_1,\hdots,x_n$ of $\BC^n$  and 
write $A$ in the form $$A=\sum \tilde{a}_{i,j}x_ix_j^*.$$ (Hence the $ \tilde{a}_{i,j}$ form the matrix obtained from $A$ by a basis change to $x_1,\hdots,x_n$.) Define the four blocks
 $ A_1=\sum_{i,j\leq r}\tilde{a}_{i,j}x_ix_j^*$, $ A_2=\sum_{i\leq r, j> r}\tilde{a}_{i,j}x_ix_j^* $, $ A_3=\sum_{i>r, j\leq r}\tilde{a}_{i,j}x_ix_j^*=A_2^* $ and $ A_4=\sum_{i,j>r}\tilde{a}_{i,j}x_ix_j^*$. It is well known that $\partial \Vert X \Vert_1$ consists of all matrices of the form
 $$
 S= \sum_{i=1}^{r}\sgn(\lambda_i) x_i x_i^*+S_2,
  $$ where $S_2\in\CH_n$ has the property that $S_2x_i=0$ for all $i\in \{1,\hdots,r \}$ and $\Vert S_2 \Vert_{\infty}\leq 1$. (See for example \cite{wa92}, where the real analogue is shown.) 
Consider now 
$$
S= \sum_{i=1}^{r}\sgn(\lambda_i) x_i x_i^*+ \tau^{-1}A_4\in \partial \Vert X \Vert_1,
\quad \mbox{ where } \tau = \|A_4\|_\infty.
$$ 
(If $\tau=0$, let $S= \sum_{i=1}^{r}\sgn(\lambda_i) x_i x_i^*$.) 
To simplify the notation, write $S_1=\sum_{i=1}^{r}\sgn(\lambda_i) x_i x_i^*$. Then 
\begin{align*}
 \Vert A- \tau S\Vert_2&= \Vert A- A_4-\tau S_1\Vert_2=\left( \tr(A_1-\tau S_1)^2+2\tr(A_2^*A_2)\right )^{1/2} \\
& =\left(\Vert A_1-\tau S_1) \Vert_2^2+2 \Vert {A}_2^*\Vert^2_2\right )^{1/2} \leq \left( 2\Vert {A}_1\Vert^2_2+2\Vert\tau\cdot  S_1 \Vert^2_2+2 \Vert {A}_2^*\Vert^2_2\right)^{1/2}   \\
& =\left(2\Vert A \cdot x_1\Vert_2^2+\hdots+2\Vert A \cdot x_r\Vert_2^2+2\Vert\tau\cdot  S_1 \Vert^2_2\right)^{1/2}\\& \leq\left(2r\| A\|_\infty^2+2r\tau^2\right)^{1/2}\leq 2\sqrt{r} \|A\|_\infty,
\end{align*} since $\tau=\Vert {A}_4\Vert_{\infty}\leq \Vert A\Vert_{\infty}=\lambda$.
\end{proof}

\subsection{Proof of Theorem~\ref{Th1}}

In order to prove both statements of Theorem~\ref{Th1}, it is enough by Proposition \ref{CRPWT} to show that for $m\geq c nr$ with probability at 
least $1-e^{-\gamma m}$ 
$$
\inf_{Y\in E_r}\left (\sum_{j=1}^m \tr(a_ja_j^*Y)^2\right )^{1/2}\geq c_1\sqrt m
$$ 
for  suitable positive constants $c,c_1,\gamma $. For $\xi>0$ let 
\begin{equation}
Q_{\xi}=\inf_{Z\in E_r}\BP(\vert\tr(a_ja_j^*Z)\vert \geq \xi). \label{eq:Q}
\end{equation}
Further let 
\begin{equation}
H=\frac{1}{\sqrt m}\sum_{j=1}^m\varepsilon_ja_ja_j^*, \label{eq:H}
\end{equation}
where the $\varepsilon_j$  form a Rademacher sequence independent of everything else, and
introduce 
$$
W_m=\BE \sup_{Y\in E_r} \tr(H\cdot Y ).
$$ 
By Theorem \ref{KMT}, for any $\xi>0$ and any $t\geq 0$ with probability at least $1-e^{-2t^2},$
$$
\inf_{Y\in E_r}\left (\sum_{j=1}^m (\tr(a_ja_j^*Y))^2\right)^{1/2}\geq \xi\sqrt m Q_{2\xi} -2W_m-\xi t.
$$ 
Following Tropp's bowling scheme, we first estimate $Q_{2\xi}$ for a suitable $\xi$.
As  in \cite{tr14}, we conclude from the Payley-Zygmund inequality (see e.g.\ \cite[Lemma~7.16]{fora13}) that 
\begin{equation}\label{PZ}
\mathbb P \{\vert \langle {aa^*},U\rangle \vert^2\geq \frac{1}{2}(\mathbb E \vert \langle {aa^*},U\rangle \vert^2) \}\geq \frac{1}{4}\cdot \frac{(\mathbb E \vert \langle {aa^*},U \rangle \vert^2)^2 }{\mathbb E \vert \langle {aa^*},U \rangle \vert^4 }.
\end{equation} (Here $a$ follows the standard Gaussian distribution on $\BC^n$.)
Assume now $\Vert U \Vert_2=1$ and write $U=\sum_{i}\lambda_i u_iu_i^*$, where $\sum_i \lambda_i^2=1$ and the $u_i$ are orthonormal.
Then $\langle {aa^*},U \rangle= \tr(aa^*U)=\sum_j \lambda_j \tr(aa^*u_ju_j^*)=\sum_j \lambda_j \vert u_j^*a\vert^2 $ and hence,
$$
\vert \langle {aa^*},U \rangle \vert^2=\sum_{i,j} \lambda_i\lambda_j \vert u_i^*a\vert^2\vert u_j^*a\vert^2.
$$
The $u_j^*a$ form independent standard (complex) Gaussian random variables. To compute the moments of  a standard complex Gaussian  random variable $Z$, write $Z=X+iY$ where $X,Y$ are independent and $\CN(0,\frac{1}{2})$ distributed. The $2k$-th moment of $X$ resp.\ $Y$ is $\frac{(2k)!}{2^{2k}k!}$, which allows us to compute higher moment of $Z$, for example, $\BE \vert Z\vert^2=\BE X^2+\BE Y^2=1$ and $\BE \vert Z \vert^4=\BE X^4+2\BE X^2\BE Y^2+ \BE Y^4=2.$
Similarly, we obtain $\BE \vert Z \vert^6=6$ and $\BE \vert Z \vert^8=24$ (and more generally $\BE \vert Z \vert^{2k}=k!$).
Thus, we conclude that
 \begin{equation}\label{GM2}
\mathbb E\vert \langle {aa^*},U \rangle \vert^2=\sum_{i\neq j}\lambda_i\lambda_j+2\sum_{i}\lambda_i^2= \sum_{i,j} \lambda_i\lambda_j+\sum_i \lambda_i^2=(\sum_i \lambda_i)^2+1\geq 1
\end{equation}
and $$
(\mathbb E\vert \langle aa^*,U \rangle \vert^2)^2=(\sum_i \lambda_i)^4+2(\sum_i \lambda_i)^2+ 1.
$$ 
Expanding 
$\mathbb E\vert \langle {aa^*},U \rangle \vert^4$ in a similar way, we obtain 
\begin{align*}
\mathbb E\vert \langle {aa^*},U \rangle \vert^4&=\sum_{i,j,k,\ell}\lambda_i\lambda_j\lambda_k\lambda_\ell+\sum_{i,k,\ell}\lambda_i^2\lambda_k\lambda_\ell+2\sum_{i,k}\lambda_i^2\lambda_k^2+4\sum_{i,k}\lambda_i^3\lambda_k+16\sum_i \lambda_i^4\\
&=(\sum_i \lambda_i)^4+(\sum_i \lambda_i)^2+2+4(\sum_i \lambda_i)(\sum_i \lambda_i^3)+16\sum_i \lambda_i^4,
\end{align*} 
where we used that $\sum_i \lambda_i^2=1$.
Again because of $\sum_i \lambda_i^2=1$ we have $\vert \lambda_i\vert \leq 1$ for all $i$ and hence $\vert \sum_i \lambda_i^3\vert \leq \sum_i \lambda_i^2=1$ and similarly $\sum_i \lambda_i^4\leq\sum_i \lambda_i^2=1$.
Also observe that $\vert \sum_i \lambda_i \vert \leq 1+(\sum_i \lambda_i)^2$.
Combining these inequalities with the above expressions for $\mathbb E\vert \langle {aa^*},U \rangle \vert^4$ and $(\mathbb E\vert \langle aa^*,U \rangle \vert^2)^2$, we obtain the inequality 
$$
\mathbb E\vert \langle {aa^*},U \rangle \vert^4\leq 24(\mathbb E\vert \langle aa^*,U \rangle \vert^2)^2.
$$
Combining this with (\ref{PZ}) and (\ref{GM2}), we obtain $$
Q_{1/\sqrt 2}\geq \frac{1}{96}.
$$
Thus we choose  $\xi=\frac{1}{2\sqrt 2}$.

In order to estimate $W_m$, we use Lemma \ref{absch} to obtain
\begin{equation}
W_m=\BE \sup_{Y\in E_r} \tr(H\cdot Y )\leq 2\sqrt{r}\cdot \BE \Vert H\Vert_{\infty}. \label{eq:Wm}
\end{equation}
By the arguments in \cite[Section~5.4.1]{ve12} we have $\BE \Vert H\Vert_{\infty}\leq c_2 \sqrt n$ if $m\geq c_3 n$ for suitable constants $c_2, c_3$, 
see also \cite[Section~8]{tr14}. Choosing $t=c_4\sqrt m$ and $m \geq cnr$ for suitable constants $c,c_4$, the proof of Theorem~\ref{Th1} is completed.

\begin{remark}
In \cite{candes_solving_2012}, a uniform result for phase retrieval in the Gaussian case is proved using an inexact dual certificate. One can
write down a generalization of this dual certificate for the rank $r$-case, but following the arguments of loc. cit., the resulting number of required 
measurements then seems to depend significantly worse than linearly on $r$. It might be possible to rather adapt the arguments in \cite{grkrku13,gross_improved_2014} 
based on a different construction of a dual certificate in order to derive linear scaling of $m$ in $r$, but the resulting proof would be more complicated than ours (and likely
lead to more logarithmic factors).
 \end{remark}

\subsection{Proof of Theorem \ref{Th2}}

Let us now turn to proving the analogous result for complex projective 4-designs. 
It is convenient to rescale the (normalized) 
4-design vectors as
\begin{equation}
\tilde{w}_i := \sqrt[4]{(n+1)n} \; w_i \quad \forall i=1,\ldots,N. \label{eq:design_normalization}
\end{equation}
This mimics the expected length of random Gaussian vectors (which corresponds to $\BE \| a_j \|_2^2 = n$) and 
we will call the system $\{\tilde{w}_i\}$ a {\em super-normalized} $4$-design.
We can apply the same technique as in the proof of Theorem \ref{Th1}, provided that we can derive a suitable lower bound for $Q_{2\xi}$ for some $0<\xi <1/2$ and an upper bound for $\BE \| H \|_{\infty}$. 
The following two technical propositions serve this purpose.

\begin{proposition} \label{prop:designs_Q}
Assume that $a$ is drawn at random from a super-normalized 
weighted $4$-design.
Then \begin{equation}
Q_{\xi} = \inf_{Z \in E_r} \BP \left( | \tr \left( a a^* Z \right)| \geq \xi \right) \geq \frac{(1-\xi^2)^2}{24} \label{eq:designs_Q}
\end{equation}
for all $\xi \in [0,1]$. 
\end{proposition}

The proof of this statement is similar to the proof of Theorem 4 in \cite{ambainis_quantum_2007} and -- likewise -- equation (15) in \cite{matthews_distinguishability_2009}. However, since we are interested in a bound on the probability of an event happening, rather than bounding an expectation value, we use the Payley-Zygmund inequality instead of Berger's one \cite{berger_fourth_1997} (which states $\BE \left[ |S | \right] \geq \BE \left[ S^2 \right]^{3/2}\BE \left[ S^4 \right]^{-1/2}$). 

\begin{proof}
The desired statement follows, if we can show that
\begin{equation}
\BP \left( | \tr \left( a a^* Z \right) | \geq \xi \right) \geq \frac{(1-\xi^2)^2}{24} \label{eq:Qaux2}
\end{equation}
holds for any matrix $Z \in \CH_n$ obeying $\| Z \|_2 =1$. 
For such  $Z$ we define the random variable $S := | \tr \left( a a^* Z \right) |$. 
Since $a$ is chosen at random from a (super-normalized) complex projective 4-design, we can use the design's defining property (\ref{eq:exact_designs}) together with (\ref{eq:PSym_Haar}) to evaluate the second and fourth moment of $S$. 
Indeed,
\begin{eqnarray*}
\BE S^2 
&=& \BE \tr \left( a a^* Z \right)^2 
= \tr \left( \BE \left( a a^* \right)^{\otimes 2} Z^{\otimes 2} \right) 
= \tr \left( \sum_{i=1}^N p_i \left( \tilde{w}_i \tilde{w}_i^* \right)^{\otimes 2} Z^{\otimes 2} \right) \\
&=& (n+1)n  \; \tr \left( \sum_{i=1}^N p_i \left( w_i w_i^* \right)^{\otimes 2} Z^{\otimes 2} \right) 
= (n+1)n \binom{n+1}{2}^{-1} \tr \left( P_{\Sym^2} Z^{\otimes 2} \right) \\
&=& 2 \tr \left( P_{\Sym^2} Z^{\otimes 2} \right) 
\end{eqnarray*}
and likewise
\begin{eqnarray*}
\BE S^4
&=& \BE \tr \left( a a^* Z \right)^4
=  \tr \left( \sum_{i=1}^N p_i \left( \tilde{w}_i \tilde{w}_i^* \right)^{\otimes 4} Z^{\otimes 4} \right) 
= \frac{4! (n+1) n}{(n+3)(n+2)} \tr \left( P_{\Sym^4} Z^{\otimes 4} \right).
\end{eqnarray*}
The remaining right hand sides are standard expressions in multilinear algebra and can for instance be calculated using wiring calculus. Indeed, 
Lemma \ref{lem:Psym}  in the appendix implies that
\begin{equation}
\BE S^2 = 2 \tr \left( P_{\Sym^2} Z^{\otimes 2} \right) = \tr (Z)^2 + \tr (Z^2) = \tr (Z)^2 + 1, \label{eq:Qaux34}
\end{equation}
because $\tr (Z^2) = \| Z \|_F^2 = 1$ by assumption, hence,
\[
(\BE S^2)^2 \geq \max\{1, \tr(Z)^4\}.
\]
Similarly, Lemma \ref{lem:Psym} assures
\begin{eqnarray*}
 \BE S^4
&=& \frac{4! (n+1) n}{(n+3)(n+2)} \tr \left( P_{\Sym^4} Z^{\otimes 4} \right) \\
&=& \frac{(n+1)n}{(n+3)(n+2)} \left( 6 \tr (Z^4) + 8 \tr (Z) \tr (Z^3) + 6 \tr (Z)^2 \tr (Z^2) + 3 \tr (Z^2)^2 + \tr (Z)^4 \right) \\
&\leq&  \left( 6 \tr (Z^4) + 8 \tr (Z) \tr (Z^3) + 6 \tr (Z)^2  + \tr (Z)^4 + 3 \right),
\end{eqnarray*}
where the simplifications in the last line are due to $\tr (Z^2) = \| Z \|_F^2 = 1$ and $\frac{(n+1)n}{(n+3)(n+2)} \leq 1$. 
Using the hierarchy of Schatten-$p$-norms -- in particular $\tr (Z^4) = \| Z \|_4^4 \leq \| Z \|_2^4=1$ and $\tr (Z^3) \leq \| Z \|_3^3 \leq \| Z \|_2^3 = 1$ -- yields
\begin{align*}
 \BE S^4 & \leq 6 \tr (Z^4) + 8 \tr (Z) \tr (Z^3) + 6 \tr (Z)^2  + \tr (Z)^4 + 3 \\
& \leq  \left( 6\| Z \|_4^4 + 8 \| Z \|_3^3 +10 \right) \max \left\{ 1, \tr (Z)^4 \right\} \leq  24 \max \left\{ 1, \tr (Z)^4 \right\}.
\end{align*}
Having precise knowledge of the second and fourth moments and the trivial fact that $\tr (Z)^2 \geq 0$  allows us to use the Payley-Zygmund inequality (for the random variable $S^2$) to bound
\begin{align*}
\BP \left( | \tr \left( a a^* Z \right) | \geq \xi \right)
& =  \BP \left( S^2 \geq \xi^2 \right) 
\geq \BP \left( S^2 \geq \xi^2 \left( 1 + \tr (Z)^2 \right) \right) \\
&=  \BP \left( S^2 \geq \xi^2 \BE S^2  \right) 
 \geq  \left( 1 - \xi^2 \right)^2 \frac{ (\BE {S^2})^2}{\BE S^4} \\
 & \geq (1-\xi^2)^2 \frac{\max\{1,\tr(Z)^4\}}{24 \max\{1,\tr(Z)^4\}} = \frac{(1-\xi^2)^2}{24}.
\end{align*}
This completes the proof.
\end{proof}

\begin{proposition} \label{prop:designs_H}
Let $H$ be the matrix defined in (\ref{eq:H}), where the $a_j$'s are chosen independently at random from a super-normalized 
weighted 1-design. 
Then it holds that
\begin{equation}
\BE \| H \|_{\infty} \leq c_4 \sqrt{n \log (2n)} \quad \mbox{with }c_4 = 3.1049, \label{eq:designs_H}
\end{equation}
provided that $m \geq {2n \log n}$. 
\end{proposition}

\begin{proof}
Since the $\epsilon_j$'s in the definition of $H$ form a Rademacher sequence, the non-commutative Khintchine inequality \cite[p.~19]{ve12}, 
see also \cite[Exercise~8.6(d)]{fora13},
is applicable and yields
\begin{align}
\BE \| H \|_{\infty} 
&= \BE_a \BE_\epsilon \frac{1}{\sqrt{m}} \left\| \sum_{j=1}^m \epsilon_j a_j a_j^* \right\|_{\infty}
\leq \sqrt{\frac{2\log (2n)}{m}} \BE_a \left\| \left( \sum_{j=1}^m \left( a_j a_j^*\right)^2 \right)^{1/2} \right\|_{\infty} \nonumber\\
&= \sqrt{\frac{2 \log (2n)}{m}} \BE_a \left\| \sqrt{(n+1)n} \sum_{j=1}^m a_j a_j^* \right\|_{\infty}^{1/2}
 \leq  \sqrt{\frac{2 \sqrt{2} n \log (2n)}{m}} \left( \BE_a \left\| \sum_{j=1}^m a_j a_j^* \right\|_{\infty} \right)^{1/2} \label{eq:Haux1}.
\end{align}
Here we have used super-normalization of our design vectors $(a_j a_j^*)^2 = \| a_j \|_2^2 a_j a_j^* = \sqrt{(n+1)n} a_j a_j^*$ 
according to (\ref{eq:design_normalization}),
the fact that $\| Z^{1/2} \|_{\infty} = \| Z \|_{\infty}^{1/2}$ holds for $Z \in \mathcal{H}_d$ arbitrary and Jensen's inequality in the last estimate. 
It remains to bound $\BE \| \sum_{j} a_j a_j^* \|_{\infty}$.
To this end, we will use the matrix Chernoff inequality of Theorem~\ref{thm:matrix:Chernoff} for $X_j = a_j a_j^*$ and calculate
\begin{eqnarray}
\| X_j\|_{\infty}
&=& \| a_j a_j^* \|_{\infty} 
= \| a_j \|_2^2 \leq \max_{1 \leq i \leq N} \|\tilde{w}_i \|_2^2 = \sqrt{(n+1)n} \leq \sqrt{2} n =: R, 
\label{eq:matrix_Chernoff1}  \\
\| \sum_{j=1}^m \BE X_j \|_{\infty}
&=& \| \sum_{j=1}^m  \sum_{i=1}^N   p_i \tilde{w}_i \tilde{w}_i^* \|_{\infty} 
= m \sqrt{n(n+1)} \left\| \sum_{i=1}^N p_i w_i w_i^*  \right\|_\infty\nonumber \\
&=&m \sqrt{(n+1)n} \left\| \frac{1}{n} \id \right\|_\infty = \frac{m \sqrt{(n+1)n}}{n} \leq \sqrt{2}m, 
\label{eq:matrix_Chernoff2}
\end{eqnarray}
where we once more have taken into account super-normalization and used the 1-design property. 
Theorem~\ref{thm:matrix:Chernoff} together with the assumption $m \geq2 n \log n$ implies that, for any $\tau > 0$,
\begin{align*}
\BE \| \sum_{j=1}^m a_j a_j^* \|_{\infty} 
& \leq \frac{e^\tau-1}{\tau} \sqrt{2} m + \tau^{-1} \sqrt{2} n \log(n) \leq
 \frac{e^\tau-1}{\tau} \sqrt{2} m + \tau^{-1} \sqrt{2} m/2\\
 & = \left(\frac{e^{\tau}-1}{\tau} \sqrt{2} + \frac{1}{\sqrt{2}\tau}\right) m.
\end{align*}
The choice $\tau = 1.27$ approximately minimizes the above expression and yields 
\[
\BE \| \sum_{j=1}^m a_j a_j^* \|_{\infty} \leq c_5 m \quad \mbox{ with } c_5 = 3.4084.
\]
Combining this estimate with (\ref{eq:Haux1}) yields the desired statement with $c_4 = 2^{3/4} \sqrt{c_5} = 3.1049$.
\end{proof}

Now we are ready to prove the second main theorem of this work.

\begin{proof}[Proof of Theorem \ref{Th2}] The proof of Theorem~\ref{Th1} shows that we only need suitable bounds for $Q_{2 \xi}$  and  for  $ \BE \Vert H\Vert_{\infty}$  (both notions are defined analogously to the Gaussian case). Fix $0 < \xi < 1/2$ arbitrary. 
For any such $\xi$, a lower bound for  $Q_{2 \xi}$  is provided by Proposition \ref{prop:designs_Q} and an upper bound for $ \BE \Vert H\Vert_{\infty}$ in this case  can be obtained from Proposition \ref{prop:designs_H}. Setting $m= C_4 n r \log n$, choosing  the constants $C_4,C_5$ and $C_6$ appropriately (depending on the particular choice of $\xi$) and applying Theorem~\ref{KMT} then yields the desired result in complete analogy to the Gaussian case (proof of Theorem \ref{Th1}). 
\end{proof}

\begin{remark}
The difference in the sampling rate $m$ by a factor proportional to $\log n$ in Theorems~\ref{Th1} and \ref{Th2} stems from the fact that Proposition~\ref{prop:designs_H} is by a factor of $\sqrt{\log (n)}$ weaker than its Gaussian analogue \cite[Section 5.4.1]{ve12}, where $\BE \| H \|_\infty \leq c_2 \sqrt{n}$. 
\end{remark}

\subsection{Proof of Theorem~\ref{Th1} for real Gaussian vectors} \label{sub:real_Gauss_proof}

As already mentioned in paragraph \ref{subsub:real_Gauss} the proof of this statement is almost identical to the proof of Theorem \ref{Th1}.
The only difference is the estimate of $Q_{2\xi}$. Using the moments of the  real instead of the complex standard Gaussian distribution, 
the reasoning in the proof of Theorem \ref{Th1} yields the estimates $\mathbb E\vert \langle aa^*,U \rangle \vert^2\geq 2$,
(compare also with \cite{tr14}). Using  real moments, one further obtains
$
\mathbb E\vert \langle {aa^*},U \rangle \vert^4\leq 27(\mathbb E\vert \langle aa^*,U \rangle \vert^2)^2
$
(alternatively one can use Gaussian hypercontractivity as done in \cite{tr14}, which gives the factor $81$ instead of $27$.) This yields
$Q_{1}\geq \frac{1}{108}$,
and the rest of the proof is the same as before.

\subsection{Proof for recovery of positive semidefinite matrices}

The only part in the proof of the recovery result for positive semidefinite matrices stated in Section~\ref{subsub:psd} 
that slightly differs from the one for arbitrary Hermitian matrices, is the proof of a corresponding version of Lemma \ref{absch}.
The subdifferential of the function $f$ introduced in \eqref{f:posdef} slightly differs from the subdifferential of the nuclear norm.
For $X=\sum_{i=1}^{r}\lambda_i x_i x_i^* $, where all $\lambda_i$ are nonzero, 
 $\partial f(X)$ consists of all matrices of the form
 $$
 S= \sum_{i=1}^{r} x_i x_i^*+S_2,
  $$ where $S_2\in\CH_n$ has the property that $S_2x_i=0$ for all $i\in \{1,\hdots,r \}$ and all eigenvalues of $S_2$ do not exceed $1$. 
  Hence we choose (in the notation of the proof of Lemma \ref{absch})
  $$S= \sum_{i=1}^{r} x_i x_i^*+ \tau^{-1}A_4\in \partial f(X).$$ Then the remainder of the proof of Lemma~\ref{absch} is the same. 

\subsection{Proof of Theorem~\ref{Th3}}

The proof of this generalized statement proceeds along the same lines as the one of Theorem~\ref{Th2}. 
However, Propositions~\ref{prop:designs_Q} and \ref{prop:designs_H} -- as well as their respective proofs -- have to be slightly altered due to the weaker requirements imposed by Theorem \ref{Th3}. 

\subsubsection{Generalized version of Proposition \ref{prop:designs_Q}}

Under the assumptions of Theorem \ref{Th3}, a weaker version of (\ref{eq:designs_Q}), namely
\begin{equation}
Q_{\xi} = \inf_{Z \in E_r} \BP \left( | \tr \left( a a^* Z \right)| \geq \xi \right) \geq \frac{(1-2\xi^2)^2}{192} \label{eq:generalized_designs_Q}
\end{equation}
for all $0 \leq \xi \leq 1/\sqrt{2}$ is still valid. This statement can be shown analogously to Proposition~\ref{prop:designs_Q}. However, one has to establish bounds on the second and fourth moments in a slightly more involved way, 
depending also on the type of design accuracy. 
Let us start with generalizing the second moment estimate  of $S := | \tr \left( a a^* Z \right) |$ for an approximate $4$-design with operator 
norm accuracy $\theta_\infty \leq 1/(16r^2)$:
\begin{align}
\BE S^2 
&= (n+1) n \left( \sum_{i=1}^N p_i \left( w_i w_i^* \right)^{\otimes 2}, Z^{\otimes 2} \right) \nonumber \\
&= 2 \left( P_{\Sym^2}, Z^{\otimes 2} \right) + (n+1)n \left( \sum_{i=1}^N p_i \left( w_i w_i^* \right)^{\otimes 2} - \binom{n+1}{2}^{-1} P_{\Sym^2}, Z^{\otimes 2} \right) \nonumber \\
& \geq  2 | \left( P_{\Sym^2}, Z^{\otimes 2} \right)| - (n+1)n \left\| \sum_{i=1}^N p_i \left( w_i w_i^* \right)^{\otimes 2} - \binom{n+1}{2}^{-1} P_{\Sym^2} \right\|_{\infty} \left\| Z^{\otimes 2} \right\|_1 \label{eq:generalization_aux1} \\
& \geq  2 | \left( P_{\Sym^2}, Z^{\otimes 2} \right)| - 2 \theta_\infty \| Z \|_1^2 
\geq 2 |\left( P_{\Sym^2}, Z^{\otimes 2} \right)| - \frac{8 r}{16 r^2} , \nonumber \\
& >  2 |\left( P_{\Sym^2}, Z^{\otimes 2} \right)| - 1/2, \label{eq:generalization_aux2}
\end{align}
where we have used the fact that $ \left( P_{\Sym^2}, Z^{\otimes 2} \right)= |\left( P_{\Sym^2}, Z^{\otimes 2} \right)|$ (see Lemma~\ref{lem:Psym}), the matrix H{\"o}lder inequality and the fact that  $\| Z \|_1 \leq 2 \sqrt{r}$ 
-- see (\ref{eq:absch1}).
The estimates for designs with nuclear norm accuracy $\theta_1 \leq 1/4$ is very similar. 
Replacing the matrix H{\"o}lder inequality in (\ref{eq:generalization_aux1}) by
\begin{equation*}
\left( \sum_{i=1}^N p_i \left( w_i w_i^* \right)^{\otimes 2} - \binom{n+1}{2}^{-1} P_{\Sym^2} , Z^{\otimes 2} \right) \geq -\left\| \sum_{i=1}^N p_i \left( w_i w_i^* \right)^{\otimes 2} - \binom{n+1}{2}^{-1} P_{\Sym^2} \right\|_{1} \left\| Z^{\otimes 2} \right\|_\infty
\end{equation*}
yields the same lower bound (\ref{eq:generalization_aux2}) due to $ \| Z^{\otimes 2} \|_\infty = \| Z\|_\infty^2 \leq \| Z \|_2^2 = 1$ (where the last equality follows from $Z \in E_r$). 
Applying Lemma~\ref{lem:Psym} then yields
\begin{equation*}
\BE S^2 
\geq {\tr \left( Z \right)^2} + 1/2 \quad \mbox{ and } \quad  (\BE S^2)^2 \geq \frac{1}{4} \max\{1, \tr(Z)^4\}.
\end{equation*}
which is the (slightly weaker) analogue of (\ref{eq:Qaux34}). Likewise we derive a fourth moment bound:
\begin{align*}
 \BE S^4  
&= \left( \BE \left[ (a a^*)^{\otimes 4} \right] , Z^{\otimes 4} \right) 
= (n+1)^2 n^2 \left( \sum_{i=1}^N p_i \left( w_i w_i^* \right)^{\otimes 4}, Z^{\otimes 4} \right) \\
& \leq  (n+1)^2 n^2 \binom{n+3}{4}^{-1} | \left( P_{\Sym^4}, Z^{\otimes 4} \right) | \\ 
&+ (n+1)^2 n^2 \left\| \sum_{i=1}^N p_i \left( w_i w_i^* \right)^{\otimes 4} - \binom{n+3}{4}^{-1} P_{\Sym^4} \right\|_\infty \left\| Z^{\otimes 4} \right\|_1 \\
& \leq  \frac{4!(n+1)n}{(n+3)(n+2)} \left( | \left( P_{\Sym^4}, Z^{\otimes 4} \right) + \theta_\infty \| Z \|_1^4 \right) 
\leq | 4! \left( P_{\Sym^4}, Z^{\otimes 4} \right) | + 4! \frac{16 r^2}{16 r^2}.
\end{align*}
As above, using the nuclear norm accuracy $\theta_1 \leq 1/4$ instead of the operator norm accuracy yields the bound  
$\BE \left[ S^4 \right] \leq  | 4! \left( P_{\Sym^4}, Z^{\otimes 4} \right) | + 4!/4 <  | 4! \left( P_{\Sym^4}, Z^{\otimes 4} \right) | + 4!$.
Lemma ~\ref{lem:Psym} yields then in both cases
\begin{align*}
\BE \left[ S^4 \right] &\leq | 4! \tr \left( P_{\Sym^4} Z^{\otimes 4} \right) | + 24 \leq 6 \tr (Z^4) + 8 | \tr (Z) \tr (Z^3) | + 6 \tr (Z)^2 + \tr (Z)^4 + 27\\
& \leq 48 \max\{1, \tr(Z)^4\}, 
\end{align*} compare the proof of Proposition \ref{prop:designs_Q}.
Having these bounds at hand, allows for applying the Payley Zygmund inequality to obtain
\begin{align*}
\BP \left( | \tr \left( a a^* Z \right) | \geq \xi \right)
& = \BP \left( S^2 \geq \xi^2 \right) \geq \BP \left( S^2 \geq 2 \xi^2 \left( 1/2 + \tr (Z)^2 \right) \right)
\geq \BP \left( S^2 \geq 2 \xi^2 \BE \left[ S^2 \right] \right) \\
& \geq (1-2 \xi^2)^2 \frac{(\BE S^2)^2}{\BE S^4} \geq (1-2 \xi^2)^2 \frac{\max\{1,\tr(Z)^4\}/4}{48 \max\{1,\tr(Z)^4\}}
= \frac{(1-2\xi^2)^2}{192}.
\end{align*}
The proof is completed.

\subsubsection{Generalized version of Proposition \ref{prop:designs_H}}

The assumptions in Theorem \ref{Th3} assure that (\ref{eq:designs_H}) is still valid, possibly with a larger absolute constant $c_4$. 
Again, the proof of this generalized statement is very similar to the proof of Proposition~\ref{prop:designs_H}. 
Indeed, only the bound  (\ref{eq:matrix_Chernoff2}) for the matrix Chernoff inequality needs to be slightly altered. 
The assumption (\ref{eq:approx_tight_frame}) implies that
\begin{align*}
\| \sum_{j=1}^m \BE \left[ X_j \right] \|_\infty
&\leq m \sqrt{(n+1)n} \left( \| \frac{1}{n} \id \|_\infty + \| \sum_{i=1}^N p_i w_i w_i^* - \frac{1}{n} \id \|_\infty \right) 
 \leq  2 \sqrt{2}m. 
\end{align*}
Consequently, applying the matrix Chernoff inequality yields (\ref{eq:designs_H}) with a slightly larger absolute constant $c_4$.

\section{Appendix}

\subsection{Schatten $p$-norms}
\label{sec:Schatten}

Recall from Section~\ref{sec:lowrankrec} that for  $1 \leq p < \infty$ , the Schatten-$p$-norm on $\CH_n$ is defined as
\begin{equation*}
\| Z \|_p = \tr \left( |Z|^p \right)^{1/p} = \left( \sum_{i=1}^n | \lambda_i |^p \right)^{1/p},
\end{equation*}
where $\lambda_1,\ldots,\lambda_n$ denote the $n$ eigenvalues of $Z \in \CH_n$.
For  $ p= \infty$ one defines similarly 
$$
\| Z \|_{\infty}=  \max\{\vert \lambda_1\vert, \ldots, \vert \lambda_n\vert\},
$$
i.e., $ \| Z \|_{\infty}$ is the spectral norm of $Z$.
The Frobenius norm $\|\cdot\|_F = \|\cdot \|_2$ is induced by the
the Hilbert-Schmitt (or Frobenius) scalar product
\begin{equation*}
\left( X, Y \right) = \tr \left( X Y \right),
\end{equation*}
which makes $\CH_n$ a Hilbert space. 
The Schatten-$p$ norms are non-increasing in $p$, i.e. for any $0 < p \leq p' \leq \infty$
\begin{equation}
\| Z \|_p \geq \| Z \|_{p'} \label{eq:norm_hierarchy}
\end{equation}
holds for all $Z \in \CH_n$. The following relations provide converse inequalities for particular instances of Schatten $p$-norms 
that are used frequently in our work:
\begin{equation}
\| Z \|_1 \leq \sqrt{\rank (Z)} \| Z \|_2 \quad \textrm{and} \quad \| Z \|_2 \leq \sqrt{\rank (Z)} \| Z \|_\infty \quad \mbox{ for all } Z \in \CH_n. \label{eq:norm_equivalences}
\end{equation}
In addition, we often use a particular instance of the matrix H{\"o}lder inequality, namely
\begin{equation}
| \left( X,Y \right) | \leq \| X \|_1 \| Y \|_\infty \quad \mbox{ for all } X,Y \in \CH_n. \label{eq:hoelder}
\end{equation}

\subsection{Matrix Chernoff inequality}

The matrix version of the classical Chernoff inequality for the expection of a sum of independent random matrices 
shown in \cite[Theorem 5.1.1]{tr12-2} (see also \cite{tr12}) reads as follows.

\begin{theorem}\label{thm:matrix:Chernoff} Let $X_1,\hdots,X_m$ be a sequence of independent random positive definite matrices in $\CH_n$ satisfying
\[
\|X_\ell\|_\infty \leq L \quad \mbox{ almost surely for all } \ell=1,\hdots,m.
\]
Then, for any $\tau > 0$, their sum obeys
\[
\BE \| \sum_{\ell=1}^m X_\ell \|_\infty \leq \frac{e^\tau-1}{\tau} \| \sum_{\ell=1}^m \BE X_\ell \|_\infty + \tau^{-1} L \log n.
\]
\end{theorem}

\subsection{Multilinear algebra}

We briefly repeat some standard concepts in multilinear algebra which are convenient for our proof of Proposition \ref{prop:designs_Q}. 
They can be found in any textbook on multilinear algebra -- e.g. \cite{landsberg_tensors_2012} -- but we nonetheless include them here for the sake of being self-contained. 

Let $V_1,\ldots,V_k$ be (finite dimensional, complex) vector spaces and let $V_1^*,\ldots,V_k^*$ denote their duals. A function
$f: V_1 \times \cdots \times V_k \to \BC$
is \emph{multilinear}, if it is linear in each space $V_i$. We denote the space of such functions by $V_1^* \otimes \cdots \otimes V_k^*$ and call it the \emph{tensor product} of $V_1^*,\ldots,V_k^*$. 
Consequently, for one fixed $n$-dimensional vector space $V$, the tensor product $\left(V \right)^{\otimes k} = \bigotimes_{i=1}^k V $ is the space of all multilinear functions
\begin{equation}
f: \underset{k\textrm{ times }}{\underbrace{\left( V\right)^*\times\cdots\times \left(V\right)^*}}\mapsto \BC ,		\label{eq:tensor_vector}
\end{equation}
and we call the elementary elements $z_1 \otimes \cdots \otimes z_k$ the \emph{tensor product} of the vectors $z_1, \ldots, z_k \in V$. 

With this notation, the space of linear maps $V \to V$ ($n \times n$-matrices) corresponds to the tensor product $\CM_n := V \otimes V^*$ which is spanned by $\left\{x \otimes y^*: \; x,y \in V\right\}$
-- the set of all rank-1 matrices. 
Using this tensor product description of $\CM_n$ allows for defining the (matrix) tensor product $\CM_n^{\otimes k}$ in complete analogy to above. 
We refer to its elements $Z_1 \otimes \cdots \otimes Z_k$ as the tensor product of the matrices $Z_1,\ldots,Z_k \in \CM_n$. 

On this tensor space, we define the \emph{partial trace} (over the $i$-th tensor system) to be
the natural contraction
\begin{eqnarray*}
\tr_i:\; \CM_n^{\otimes k} & \to & \CM_n^{\otimes (k-1)} \\
Z_1 \otimes \cdots \otimes Z_k & \mapsto & \tr (Z_i) Z_1 \otimes \cdots \otimes Z_{i-1} \otimes Z_{i+1} \otimes \cdots \otimes Z_{k}.
\end{eqnarray*}
The partial trace over multiple systems can then be obtained by concatenating individual traces of this form, e.g.
\begin{equation}
\tr_{i,j} = \tr_i \circ \tr_j: \CM_n^{\otimes k} \to \CM_n^{\otimes (k - 2)} \label{eq:partial_trace}
\end{equation}
for $1 \leq i < j \leq k$ arbitrary and so forth. 
A particular property of arbitrary partial traces is that they preserve positive semidefiniteness -- see e.g.\ \cite[Section 8.3.1]{nielsen_quantum_2010} or any lecture notes on quantum information theory.
If a matrix $Z \in \CM_n^{\otimes k}$ is positive semidefinite, then $\tr_i \left( Z \right) \in \CM^{\otimes (k-1)}$ is again positive semidefinite for any $1 \leq i \leq k$. 
This behavior naturally extends to multiple partial traces in the sense of (\ref{eq:partial_trace}). 
The \emph{full trace} corresponds to
\begin{eqnarray*}
\tr := \tr_{1,\ldots,k}:\; \CM_n^{\otimes k} & \to & \BC \\
Z_1 \otimes \cdots \otimes Z_k & \mapsto& \tr (Z_1) \cdots \tr (Z_k).
\end{eqnarray*}
This implies that the nuclear norm is multiplicative with respect to the tensor structure, i.e.,
\begin{equation}
\| Z_1 \otimes \cdots Z_k \|_1 = \tr \left( |Z_1| \otimes \cdots \otimes |Z_k| \right) = \tr \left( |Z_1| \right) \cdots \tr \left( |Z_k | \right) = \| Z_1 \|_1 \cdots \| Z_k \|_1 \label{eq:tensor_trace_norm}
\end{equation}
for $Z_1,\ldots,Z_k \in \CM$ arbitrary. A singular value decomposition -- see e.g. \cite[Lecture 2]{watrous_lecture_2011} -- reveals that the same is true for the operator norm, i.e.
\begin{equation}
\left\| Z_1 \otimes \cdots \otimes Z_k \right\|_\infty = \| Z_1 \|_\infty \cdots \| Z_k \|_\infty \label{eq:tensor_operator_norm}.
\end{equation}

Let us now return to the $k$-fold tensor space $V^{\otimes k}$ of $n$-dimensional complex vectors. 
We define the (symmetrizer) map 
$P_{\Sym^k}:  \left(V \right)^{\otimes k}  \to  \left( V \right)^{\otimes k} $
via their action on elementary elements:
\begin{equation}
P_{\Sym^k} \left( z_1 \otimes \cdots \otimes z_k \right) := \frac{1}{k!} \sum_{\pi \in S_k} z_{\pi (1)} \otimes \cdots \otimes z_{\pi(k)}, \label{eq:symmetrizer}
\end{equation}
where $S_k$ denotes the group of permutations of $k$ elements.
This map projects $ \left(V \right)^{\otimes k}$ onto the totally symmetric subspace $\Sym^k$ of $\left(V \right)^{\otimes k}$ whose dimension \cite[Exercise 2.6.3.5]{landsberg_tensors_2012} is
\begin{equation}
\dim \Sym^k = \binom{n+k-1}{k}		\label{eq:dimsym}.
\end{equation}

Using these basic concepts of multilinear algebra and (\ref{eq:PSym_Haar}), we can show that every approximate $t$-design is also an approximate design of lower order.

\begin{lemma} \label{lem:design_hierarchy}
Every approximate $t$-design of accuracy measured either in operator- or trace-norm is also an approximate $k$-design of the same accuracy for any $1 \leq k \leq t$. 
Furthermore the accuracies $\theta_\infty$ and $\theta_1$ are related via
\begin{equation}
\theta_\infty \leq \theta_1 \leq n^t \theta_\infty. \label{eq:accuracy_relations}
\end{equation}
\end{lemma} 

This statement is implicitly proved in \cite{ambainis_quantum_2007}, where the authors use an equivalent definition of approximate $t$-designs as averaging sets of complex polynomials of degree at most $(t,t)$. With this alternative definition, Lemma \ref{lem:design_hierarchy} follows naturally from the fact that every polynomial of degree at most $(k,k)$ with $1 \leq k \leq t$ is a particular instance of a degree-$(t,t)$-polynomial. 
Here we provide an alternative proof that uses concepts from multilinear algebra and accesses Definition~\ref{def:approx_designs} directly. 
Such a proof idea is mentioned in \cite[Section 2.2.3]{low_pseudo_2010} and we include the full argument here for the sake of being self-contained. 

\begin{proof}[Proof of Lemma \ref{lem:design_hierarchy}]
Let us start with proving the statement for the accuracy 
measured in operator norm. 
In this case, Definition \ref{def:approx_designs} is equivalent to demanding
\begin{equation}
(1- \theta_\infty) \int_{\BC P^{n-1}} \left( w w^* \right)^{\otimes t} \mathrm{d} w \leq \sum_{i=1}^N p_i \left( w_i w_i^* \right)^{\otimes t} \leq (1+ \theta_\infty) \int_{\BC P^{n-1}} \left( w w^* \right)^{\otimes t} \mathrm{d} w. \label{eq:design_hierarchy_aux1}
\end{equation}
The desired statement follows if we can show that (\ref{eq:design_hierarchy_aux1}) implies a corresponding inequality for smaller tensor powers $k$.
Fix $1 \leq k \leq t$ and note that the inequality chain (\ref{eq:design_hierarchy_aux1}) is preserved under taking arbitrary partial traces, because partial traces respect the positive semidefinite ordering. 
This in particular implies that 
\begin{eqnarray*}
(1-\theta_\infty)  \int_{\BC P^{n-1}}\tr_{1,\ldots,(t-k)} \left( \left( w w^*\right)^{\otimes t} \right) \mathrm{d}w &\leq& \sum_{i=1}^N p_i \tr_{1,\ldots,(t-k)} \left( \left( w_i w_i^* \right)^{\otimes t} \right) \\
&\leq& (1+ \theta_\infty) \int_{\BC P^{n-1}} \tr_{1,\ldots,(t-k)} \left( \left( w w^*\right)^{\otimes t} \right) \mathrm{d}w
\end{eqnarray*}
remains valid. 
Due to normalization $\|w_i \|_{\ell_2} =1$ and and since we calculate the integrals using preimages of the $w \in \BC P^{n-1}$  in the unit sphere,
these expressions can be readily calculated. Indeed,
 $$
 \tr_{1,\ldots,(t-k)} \left( \left( w_i w_i^*\right)^{\otimes t} \right) = \left( w_i w_i^* \right)^{\otimes k}  |\langle w_i, w_i \rangle |^{2(t-k)}\\
=  \left( w_i w_i^* \right)^{\otimes k} $$  
and
$$ \int_{\BC P^{n-1}} \tr_{1,\ldots,(t-k)} \left( \left( w w^*\right)^{\otimes t} \right) \mathrm{d}w 
= \int_{\BC P^{n-1}} \left( w w^* \right)^{\otimes k} | \langle w, w \rangle |^{2(t-k)} \mathrm{d}w
= \int_{\BC P ^{n-1}} \left( w w^* \right)^{\otimes k} \mathrm{d}w.$$
The desired statement follows.

The analogous statement for accuracy measured in trace-norm directly follows from the fact that the nuclear norm is monotonic with respect to partial traces, i.e., 
$ \| \tr_i (Z) \|_1 \leq \| Z \|_1$ for any $Z \in \CM_n^{\otimes t}$ and $1 \leq i \leq t$ \cite[Lecture 2]{watrous_lecture_2011}. Combining this with the calculations above reveals that
\begin{align*}
&\left\| \sum_{i=1}^N p_i \left( w_i w_i^* \right)^{\otimes k} - \int_{\BC P^{n-1}} \left( w w^* \right)^{\otimes k} \mathrm{d} w \right\|_1\\
&=\left\| \tr_{1,\ldots,t-k} \left( \sum_{i=1}^N p_i \left( w_i w_i^* \right)^{\otimes t} - \int_{\BC P^{n-1}} \left( w w^* \right)^{\otimes t} \mathrm{d} w \right) \right\|_1 \\
& \leq  \left\| \sum_{i=1}^N p_i \left( w_i w_i^* \right)^{\otimes t} - \int_{\BC P^{n-1}} \left( w w^* \right)^{\otimes t} \mathrm{d} w  \right\|_1 
\leq \theta_1.
\end{align*}
Finally, inequality (\ref{eq:accuracy_relations}) directly follows from comparing trace  and operator norm on $\CM_n^{\otimes t}$ which is isomorphic to the space of all $n^t \times n^t$-dimensional matrices.

\end{proof}

\subsection{Wiring calculus in multilinear algebra}

The defining properties (\ref{eq:exact_designs}), (\ref{eq:approx_designs}) of exact and approximate complex projective $t$-designs are phrased in terms of tensor spaces.
For calculations in multilinear algebra -- particularly if they involve (partial) traces-- \emph{wiring diagrams} \cite[Chapter 2.11]{landsberg_tensors_2012} are very useful,
as they provide a way of computing contractions of tensors pictorially.
Here we give a brief introduction that should suffice for our calculations and defer the interested reader to \cite{grkrku13} and references therein for further reading.

In wiring calculus, every tensor is associated with a box, and every index corresponds to a line emanating from this box. 
Two connected lines correspond to connected indices. 
The formalism becomes much clearer when applying it to matrix calculus.  A matrix $Z : \BC^n \to \BC^n$ can be viewed as two-index-tensors ${Z^i}_j$
and is thus represented by a node 
$
\tikz[heighttwo,xscale=.5,baseline]{
\coordinate(up)at(0,0.8);
\coordinate(mid)at(0,0.5);
\coordinate(down)at(0,0.2);
\draw(down)to(up);
\draw(mid)node[vector]{$Z$};
}
$
with upper line corresponding to the index $i$ and the lower one to $j$. 
Two matrices $Y,Z$ are multiplied by contracting $Z$'s upper index with $Y$'s lower one:
\begin{equation*}
{(Y Z )^i}_j = \sum_{k=1}^n {Y^i}_k {Z^k}_j.
\end{equation*}
In wiring calculus matrix multiplication is therefore represented by
\begin{equation*}
	YZ
	=
	\tikz[heighttwo,xscale=.5,baseline]{
	\coordinate(up)at(0,01);
	\coordinate(mid1)at(0,0.73);
	\coordinate(mid2)at(0,0.27);
	\coordinate(down)at(0,0);
	\draw(down)to(up);
	\draw(mid1)node[vector]{$Y$};
	\draw(mid2)node[vector]{$Z$};
	}.
\end{equation*}
Tensor products of matrices are arranged in parallel, i.e.,
\begin{equation*}
	Y \otimes Z
	=
	\tikz[heighttwo,xscale=.5,baseline]{
	\coordinate(up)at(0,0.9);
	\coordinate(upr)at(1,0.9);
	\coordinate(mid)at(0,0.5);
	\coordinate(midr)at(1,0.5);
	\coordinate(down)at(0,0.1);
	\coordinate(downr)at(1,0.1);
	\draw(down)to(up);
	\draw(downr)to(upr);
	\draw(mid)node[vector]{$Y$};
	\draw(midr)node[vector]{$Z$};
	}.
\end{equation*}
Taking traces of tensor products, e.g.,
\begin{equation*}
Y \otimes Z \mapsto \tr ( Y \otimes Z ) = \sum_{i,j=1}^n {Y^i}_i {Z^j}_j
\end{equation*}
just corresponds to contracting parallel matrix indices and therefore
\begin{equation*}
	\tr (Y \otimes Z) 
=
	\tikz[heighttwo,xscale=.5,baseline]{
	\coordinate(up)at(0,0.9);
	\coordinate(upr)at(1,0.9);
	\coordinate(mid)at(0,0.5);
	\coordinate(midr)at(1,0.5);
	\coordinate(down)at(0,0.1);
	\coordinate(downr)at(1,0.1);
	\coordinate(leftup)at(-0.5,0.9);
	\coordinate(leftdown)at(-0.5,0.1);
	\coordinate(rightup)at(1.5,0.9);
	\coordinate(rightdown)at(1.5,0.1);
	\draw(down)to(up);
	\draw(downr)to(upr);
	\draw(leftdown)to(leftup);
	\draw(rightdown)to(rightup);
	\draw(mid)node[vector]{$Y$};
	\draw(midr)node[vector]{$Z$};
	\draw(leftup)to[out=90,in=90](up);
	\draw(leftdown)to[out=-90,in=-90](down);
	\draw(upr)to[out=90,in=90](rightup);
	\draw(downr)to[out=-90,in=-90](rightdown);
	},
\end{equation*}
which straightforwardly extends to larger (and smaller, namely 
$
	\tr (Z) 
	= 
	\tikz[heighttwo,xscale=.5,baseline]{
	\coordinate(up)at(0,0.8);
	\coordinate(trup)at(0.5,0.8);
	\coordinate(mid)at(0,0.5);
	\coordinate(down)at(0,0.2);
	\coordinate(trdown)at(0.5,0.2);
	\draw(down)to(up);
	\draw(trdown)to(trup);
	\draw(up)to[out=90,in=90](trup);
	\draw(down)to[out=-90,in=-90](trdown);
	\draw(mid)node[vector]{$Z$};
	}
$) tensor systems. 

Finally, we are going to require \emph{transpositions} on $\left( \BC^n \right)^{\otimes t}$ which act by interchanging the $i$-th and $j$-th tensor factor.
For example
\begin{equation*}
\sigma_{(1,2)} \left( x \otimes y \otimes \cdots \right) = y \otimes x \otimes \cdots,
\end{equation*}
with $x,y \in \BC^n$ arbitrary. Note that these transpositions generate the full group of permutations. For $\left( \BC^n \right)^{\otimes 2}$ there are only two transpositions, namely
\begin{equation*}
\underline{1} = 
\tikz[heighttwo,xscale=.5,baseline]{
\coordinate(up)at(0,0.9);
\coordinate(upr)at(1,0.9);
\coordinate(mid)at(0,0.5);
\coordinate(midr)at(1,0.5);
\coordinate(down)at(0,0.1);
\coordinate(downr)at(1,0.1);
\draw(down)to(up);
\draw(downr)to(upr);
}
\;\textrm{(trivial permutation)}
\quad
\textrm{and}
\quad
\sigma_{(1,2)}
=
\tikz[heighttwo,xscale=.5,baseline]{
\coordinate(up)at(0,0.9);
\coordinate(upr)at(1,0.9);
\coordinate(mid)at(0,0.5);
\coordinate(midr)at(1,0.5);
\coordinate(down)at(0,0.1);
\coordinate(downr)at(1,0.1);
\draw(down)to[wavyup](upr);
\draw(downr)to[wavyup](up);
}.
\end{equation*} 
But for higher tensor systems more permutations can occur.
In wiring calculus, permutations therefore act by interchanging different input and output lines.

We are now ready to prove the statements required in Proposition \ref{prop:designs_Q}. 

\begin{lemma}\label{lem:Psym}
For an abritrary Hermitian matrix $Z \in \mathcal{H}_n$ and a positive integer $m$, it holds
\begin{equation*}
m! \left( P_{\Sym^m} Z^{\otimes m} \right)=\sum_{(j_1,...,j_m)\in\BN_0^m\atop  \sum_{k=1}^m k j_k=m} \frac{m!}{\prod_{k=1}^m j_k!\; k^{j_k}}\prod_{k=1}^m\tr(Z^k)^{j_k}.
\end{equation*}
In particular, for $m=2$ we obtain
\begin{equation*}
2 \tr \left( P_{\Sym^2} Z^{\otimes 2} \right) = \tr (Z)^2 + \tr (Z^2), 
\end{equation*}
and for $m=4$ we obtain
\begin{equation*}
4! \; \tr \left( P_{\Sym^4} Z^{\otimes 4} \right) = \tr (Z)^4 + 8 \tr (Z) \tr (Z^3) + 3 \tr (Z^2)^2 + 6 \tr (Z)^2 \tr (Z^2) + 6 \tr (Z^4).
\end{equation*}

\end{lemma}

\begin{proof}
We start with the case $m=2$ and then extend the argument to the general case.

\noindent The basic formula for $P_{\Sym^2}$ is given by
\begin{equation*}
P_{\Sym^2} = \frac{1}{2} \sum_{\pi \in S_2} \pi = \frac{1}{2} \left( \underline{1} + \sigma_{(1,2)} \right),
\end{equation*}
and its pictorial counterpart is therefore
\begin{equation*}
\tikz[heighttwo,xscale=.5,baseline]{
\coordinate(up)at(0,0.9);
\coordinate(upr)at(1,0.9);
\coordinate(mid)at(0,0.5);
\coordinate(midr)at(1,0.5);
\coordinate(down)at(0,0.1);
\coordinate(downr)at(1,0.1);
\draw(down)to(up);
\draw(downr)to(upr);
\draw[draw=black,fill=gray!10](-0.3,.25)rectangle(1.3,.75);\node[basiclabel]at(0.51,.48){$P_{\Sym^2}$};
}
=
\frac{1}{2} \left(
\tikz[heighttwo,xscale=.5,baseline]{
\coordinate(up)at(0,0.9);
\coordinate(upr)at(1,0.9);
\coordinate(mid)at(0,0.5);
\coordinate(midr)at(1,0.5);
\coordinate(down)at(0,0.1);
\coordinate(downr)at(1,0.1);
\draw(down)to(up);
\draw(downr)to(upr);
}
+
\tikz[heighttwo,xscale=.5,baseline]{
\coordinate(up)at(0,0.9);
\coordinate(upr)at(1,0.9);
\coordinate(mid)at(0,0.5);
\coordinate(midr)at(1,0.5);
\coordinate(down)at(0,0.1);
\coordinate(downr)at(1,0.1);
\draw(down)to[wavyup](upr);
\draw(downr)to[wavyup](up);
}
\right).
\end{equation*}
Applying the graphical calculus introduced above then yields
\begin{eqnarray*}
2 \tr \left( P_{\Sym^2} Z^{\otimes 2}\right)
&=& 
2 \;\tikz[heighttwo,xscale=.5,baseline]{
	\coordinate(top)at(0,1.3){};	\coordinate(topr)at(1,1.3){};
	\coordinate(mid1)at(0,1){};	\coordinate(mid1r)at(1,1){};
	\coordinate(mid2)at(0,0.5){};	\coordinate(mid2r)at(1,0.5){};
	\coordinate(bot)at(0,0){};		\coordinate(botr)at(1,0){};
	\coordinate(toptr)at(1.5,1.3){};	\coordinate(bottr)at(1.5,0){};
	\coordinate(ltoptr)at(-0.5,1.3){};	\coordinate(lbottr)at(-0.5,0){};
	\draw(bot)to(top);
	\draw(botr)to(topr);
	\draw(mid1)node[vector]{$Z$};
	\draw(mid1r)node[vector]{$Z$};
	\draw[draw=black,fill=gray!10](-0.3,.25)rectangle(1.3,.75);\node[basiclabel]at(0.51,.48){$P_{\Sym^2}$};
	\draw(bottr)to(toptr);
	\draw(topr)to[out=90,in=90](toptr);
	\draw(botr)to[out=-90,in=-90](bottr);
	\coordinate(ltoptr)at(-0.5,1.3){};	\coordinate(lbottr)at(-0.5,0){};
	\draw(top)to[out=90,in=90](ltoptr);
	\draw(bot)to[out=-90,in=-90](lbottr);
	\draw(lbottr)to(ltoptr);
}
=
\tikz[heighttwo,xscale=.5,baseline]{
	\coordinate(top)at(0,1.3){};	\coordinate(topr)at(1,1.3){};
	\coordinate(mid1)at(0,1){};	\coordinate(mid1r)at(1,1){};
	\coordinate(mid2)at(0,0.5){};	\coordinate(mid2r)at(1,0.5){};
	\coordinate(bot)at(0,0){};		\coordinate(botr)at(1,0){};
	\coordinate(toptr)at(1.5,1.3){};	\coordinate(bottr)at(1.5,0){};
	\coordinate(in)at(0,1.4){};
	\coordinate(out)at(0,-0.1){};
	\draw(bot)to(top);
	\draw(botr)to(topr);
	\draw(mid1)node[vector]{$Z$};
	\draw(mid1r)node[vector]{$Z$};
	\draw(bottr)to(toptr);
	\draw(topr)to[out=90,in=90](toptr);
	\draw(botr)to[out=-90,in=-90](bottr);
	\coordinate(ltoptr)at(-0.5,1.3){};	\coordinate(lbottr)at(-0.5,0){};
	\draw(top)to[out=90,in=90](ltoptr);
	\draw(bot)to[out=-90,in=-90](lbottr);
	\draw(lbottr)to(ltoptr);
}
+
\tikz[heighttwo,xscale=.5,baseline]{
	\coordinate(top)at(0,1.3){};	\coordinate(topr)at(1,1.3){};
	\coordinate(mid1)at(0,1){};	\coordinate(mid1r)at(1,1){};
	\coordinate(mid2)at(0,0.5){};	\coordinate(mid2r)at(1,0.5){};
	\coordinate(bot)at(0,0){};		\coordinate(botr)at(1,0){};
	\coordinate(toptr)at(1.5,1.3){};	\coordinate(bottr)at(1.5,0){};
	\coordinate(in)at(0,1.4){};
	\coordinate(out)at(0,-0.1){};
	\draw(bot)to[wavyup](mid1r);
	\draw(botr)to[wavyup](mid1);
	\draw(bottr)to(toptr);
	\draw(mid1r)to(topr);
	\draw(topr)to[out=90,in=90](toptr);
	\draw(botr)to[out=-90,in=-90](bottr);
	\draw(mid1)to(top);
	\draw(mid1)node[vector]{$Z$};
	\draw(mid1r)node[vector]{$Z$};
	\coordinate(ltoptr)at(-0.5,1.3){};	\coordinate(lbottr)at(-0.5,0){};
	\draw(top)to[out=90,in=90](ltoptr);
	\draw(bot)to[out=-90,in=-90](lbottr);
	\draw(lbottr)to(ltoptr);
}
=
\tikz[heighttwo,xscale=.5,baseline]{
	\coordinate(top)at(0,1.3){};	\coordinate(topr)at(1,1.3){};
	\coordinate(mid1)at(0,1){};	\coordinate(mid1r)at(1,1){};
	\coordinate(mid2)at(0,0.5){};	\coordinate(mid2r)at(1,0.5){};
	\coordinate(bot)at(0,0){};		\coordinate(botr)at(1,0){};
	\coordinate(toptr)at(1.5,1.3){};	\coordinate(bottr)at(1.5,0){};
	\coordinate(in)at(0,1.4){};
	\coordinate(out)at(0,-0.1){};
	\draw(bot)to(top);
	\draw(botr)to(topr);
	\draw(mid1)node[vector]{$Z$};
	\draw(mid1r)node[vector]{$Z$};
	\draw(bottr)to(toptr);
	\draw(topr)to[out=90,in=90](toptr);
	\draw(botr)to[out=-90,in=-90](bottr);
	\coordinate(ltoptr)at(-0.5,1.3){};	\coordinate(lbottr)at(-0.5,0){};
	\draw(top)to[out=90,in=90](ltoptr);
	\draw(bot)to[out=-90,in=-90](lbottr);
	\draw(lbottr)to(ltoptr);
}
+
\tikz[heighttwo,xscale=.5,baseline]{
\coordinate(up)at(0,1.3);
\coordinate(mid1)at(0,1);
\coordinate(mid2)at(0,0.3);
\coordinate(down)at(0,0);
\draw(down)to(up);
\draw(mid1)node[vector]{$Z$};
\draw(mid2)node[vector]{$Z$};
	\coordinate(ltoptr)at(-0.5,1.3){};	\coordinate(lbottr)at(-0.5,0){};
	\draw(top)to[out=90,in=90](ltoptr);
	\draw(bot)to[out=-90,in=-90](lbottr);
	\draw(lbottr)to(ltoptr);
}
	\\
 &=& 
\tr (Z)^2 + \tr (Z^2),
\end{eqnarray*}
which is the desired statement for $m=2$. 

Expanding $m! \left( P_{\Sym^m} Z^{\otimes m} \right)$ analogously  in the general case, we obtain for each $\pi \in S_m$ one summand which corresponds to a wiring diagram in which $m$ copies of the node $
\tikz[heighttwo,xscale=.5,baseline]{
\coordinate(up)at(0,0.8);
\coordinate(mid)at(0,0.5);
\coordinate(down)at(0,0.2);
\draw(down)to(up);
\draw(mid)node[vector]{$Z$};
}
$ are involved. More precisely, the wiring diagram corresponding to $\pi$ is obtained by connecting  for each $i\in \{1,\hdots,m\}$ the output line  of the $i$-th copy  of  $
\tikz[heighttwo,xscale=.5,baseline]{
\coordinate(up)at(0,0.8);
\coordinate(mid)at(0,0.5);
\coordinate(down)at(0,0.2);
\draw(down)to(up);
\draw(mid)node[vector]{$Z$};
}$  to the input line of the $\pi(i)$-th copy of $
\tikz[heighttwo,xscale=.5,baseline]{
\coordinate(up)at(0,0.8);
\coordinate(mid)at(0,0.5);
\coordinate(down)at(0,0.2);
\draw(down)to(up);
\draw(mid)node[vector]{$Z$};
}$. If we write $\pi$ as a product of $k$ cyclic permutations, $\pi=c_1\cdots c_k$, then the wiring diagram of $\pi$ consists of $k$ closed loops, one for each of the cyclic permutations $c_1,\hdots,c_k$. Write  $c_i=(i_1,\hdots,i_{r_i})$. Then the  loop corresponding to $c_i$ connects  $r_i$ copies of  $
\tikz[heighttwo,xscale=.5,baseline]{
\coordinate(up)at(0,0.8);
\coordinate(mid)at(0,0.5);
\coordinate(down)at(0,0.2);
\draw(down)to(up);
\draw(mid)node[vector]{$Z$};
}$. Hence the contribution of $\pi$ to the whole sum is $\tr(Z^{r_1})\cdots \tr(Z^{r_k})$. 
Thus for a given partition $m=r_1+\hdots+r_k$ of $m$, any element of $S_m$ which is the product of $k$ cyclic (and disjoint) permutations  of lengths $r_1,\hdots,r_k$ respectively  gives the same contribution  $\tr(Z^{r_1})\cdots \tr(Z^{r_k})$. 

Note that we can rewrite any partition of $m$ in the form $m=j_1\cdot 1+\hdots+j_m\cdot m$, where $j_i$ counts how often the summand $i$ appears in that partition. 
It remains to count for each partition  $m=j_1\cdot 1+\hdots+j_m\cdot m$ of $m$ how many elements of $S_m$  there are which are a product of precisely $j_1$ cyclic permutations of length $1$, of precisely  $j_2$ cyclic permutations of length $2$ and so on (all the cyclic permutations being disjoint).  It is easy to see (and well known, see for example \cite[Proposition 1.3.2]{st97})
that there are precisely $\frac{m!}{\prod_{k=1}^m j_k!\; k^{j_k}}$ such permutations in $S_m$. Each of them  contributes a summand  $\tr(Z^{1})^{j_1}\hdots \tr(Z^{m})^{j_m}$ to $m! \left( P_{\Sym^m} Z^{\otimes m} \right)$. This gives the claimed formula.
\end{proof}

\subsection*{Acknowledgements}

RK is glad to acknowledge inspiring discussions with D.\ Gross and helpful comments from J.\ Aaberg. 

The work of RK is supported by scholarship funds from the State Graduate Funding Program of Baden-W\"urttemberg,
 the Excellence Initiative of the German Federal and State Governments (Grant ZUK 43),
the ARO under contracts, W911NF-14-1-0098 and W911NF-14-1-0133 (Quantum Characterization, Verification, and Validation),
the Freiburg Research Innovation Fund, and the DFG. HR and UT acknowledge funding by the European Research Council through the Starting Grant
StG 258926 (SPALORA). RK and HR would like to thank the {\it Mathematisches Forschungsinstitut Oberwolfach} and the organizers of the 
Oberwolfach workshop {\it Mathematical Physics meets Sparse Recovery} (April 2014, Workshop ID: 1416a), where discussions on the topic of this article have started.


\end{document}